\documentclass[11pt,letterpaper]{article}

\usepackage[margin=1in]{geometry}
\usepackage{amsmath,amssymb,amsthm,mathtools,graphicx,latexsym,url,xcolor,graphicx,cite,framed,microtype,hyperref,thm-restate,xspace}
\usepackage[inline]{enumitem}
\setlist[itemize]{label=--}
\setlist[enumerate]{label=(\arabic*),labelindent=\parindent,leftmargin=*}

\usepackage{sectsty}
\allsectionsfont{\boldmath}

\hypersetup{
  colorlinks=true,
  linkcolor=[rgb]{0,0.1,0.5},
  citecolor=[rgb]{0,0.1,0.5},
  filecolor=black,
  urlcolor=[rgb]{0,0.1,0.5},
  pdftitle={Fast Approximate Shortest Paths in the Congested Clique},
  pdfauthor={Keren Censor-Hillel, Michal Dory Janne H. Korhonen and Dean Leitersdorf}
}

\newtheorem{theorem}{Theorem}
\newtheorem{claim}[theorem]{Claim}
\newtheorem{lemma}[theorem]{Lemma}

\newtheorem{question}[theorem]{Question}

\DeclareMathOperator{\nz}{nz}
\newcommand{\denssymbol}{\rho}
\newcommand{\dens}[1]{\denssymbol_{#1}}
\newcommand{\densh}[1]{\widehat{\denssymbol}_{#1}}

\newcommand{\size}[1]{\ensuremath{\left|#1\right|}}

\newcommand{\s}[1]{\overline{#1}}

\newcommand{\clique}{\textsc{Congested Clique}\xspace}
\newcommand{\knearest}{$k$\textsc{-nearest}\xspace}
\newcommand{\sdk}{$(S, d, k)$\textsc{-source detection}\xspace}
\newcommand{\distthrough}{\textsc{distance through sets}\xspace}

\begin{document}

\begin{titlepage}

\title{
    Fast Approximate Shortest Paths in the Congested Clique
}

\author{Keren Censor-Hillel\footnote{Technion, Department of Computer Science, \{ckeren,smichald\}@cs.technion.ac.il, dean.leitersdorf@gmail.com.}
\and Michal Dory\footnotemark[1] \and Janne H.\ Korhonen\footnote{IST Austria, janne.korhonen@ist.ac.at.} \and Dean Leitersdorf\footnotemark[1]}{}

\date{}	
\maketitle
\thispagestyle{empty} 
\begin{abstract}
\normalsize We design fast \emph{deterministic} algorithms for distance computation in the \clique model. Our key contributions include:
\begin{itemize}
    \item A $(2+\epsilon)$-approximation for all-pairs shortest paths in $O(\log^2{n} / \epsilon)$ rounds on unweighted undirected graphs. With a small additional additive factor, this also applies for weighted graphs. This is the first \emph{sub-polynomial} constant-factor approximation for APSP in this model.
    \item A $(1+\epsilon)$-approximation for multi-source shortest paths from $O(\sqrt{n})$ sources in $O(\log^2{n} / \epsilon)$ rounds on weighted undirected graphs. This is the first \emph{sub-polynomial} algorithm obtaining this approximation for a set of sources of polynomial size.
\end{itemize}
Our main techniques are new distance tools that are obtained via improved algorithms for sparse matrix multiplication, which we leverage to construct efficient hopsets and shortest paths.
Furthermore, our techniques extend to additional distance problems for which we improve upon the state-of-the-art, including diameter approximation, and an \emph{exact} single-source shortest paths algorithm for weighted undirected graphs in $\tilde{O}(n^{1/6})$ rounds.
\end{abstract}

\end{titlepage}


\setcounter{page}{2}
\setcounter{tocdepth}{2}
\tableofcontents
\clearpage


\section{Introduction}

Computing distances in a graph is a fundamental task widely studied in many computational settings. Notable examples are computation of all-pairs shortest paths (APSP), single-source shortest paths (SSSP), and computing specific parameters such as the diameter of a graph.
In this work, we study distance computations in the \clique model of distributed computing.

In the \clique model, we have a \emph{fully-connected} communication network of $n$ nodes, where nodes communicate by sending $O(\log{n})$-bit messages to each other node in synchronous rounds. The \clique model has been receiving much attention during the past decade or so, due to both its theoretical interest in focusing on congestion alone as a communication resource, and its relation to practical settings that use fully connected overlays~\cite{censor2015algebraic,le2016further,DBLP:conf/opodis/Censor-HillelLT18, ghaffari2018congested, lotker2005minimum, drucker2014power, DBLP:conf/podc/GhaffariP16, DBLP:conf/podc/GhaffariGKMR18, DBLP:conf/podc/Ghaffari17, DBLP:conf/icalp/Parter18, DBLP:conf/wdag/ParterS18, DBLP:conf/wdag/ParterY18, DBLP:conf/opodis/HolzerP15, hegeman2015lessons, jurdzinski2018mst, lenzen2013optimal, DBLP:conf/spaa/KorhonenS18}. In particular, there have been many recent papers studying distance problems in \clique \cite{censor2015algebraic,le2016further,DBLP:conf/opodis/Censor-HillelLT18, nanongkai2014distributed, DBLP:conf/opodis/HolzerP15,DBLP:conf/wdag/ParterY18, DBLP:conf/wdag/BeckerKKL17, DBLP:journals/corr/ElkinN17}.

\subsection{Distance computation in the congested clique}

Many state-of-the art results for distance computations in the \clique model exploit the well-known connection between computing distances and matrix multiplication \cite{censor2015algebraic,le2016further,DBLP:conf/opodis/Censor-HillelLT18}. Specifically, the $n$th power of the adjacency matrix $A$ of a graph $G = (V,E)$, taken over the \emph{min-plus} or \emph{tropical} semiring (see e.g. \cite{censor2015algebraic} for details), correspond to shortest-path distances. Hence, iteratively squaring a matrix $\log{n}$ times allows computing all the distances in the graph. This approach gives the best known algorithms for APSP in the \clique, including
(1) an $\widetilde{O}(n^{1/3})$ round algorithm for exact APSP in weighted directed graphs~\cite{censor2015algebraic},
(2) $O(n^{0.158})$ round algorithms for exact APSP in unweighted undirected graphs and $(1+o(1))$-approximate APSP in weighted directed graphs~\cite{censor2015algebraic}, as well as
(3) an $O(n^{0.2096})$ round algorithm for exact APSP in directed graphs with constant weights~\cite{le2016further}.
Additionally, in \cite{DBLP:conf/opodis/Censor-HillelLT18}, this connection is used to show an improved APSP algorithm for \emph{sparse} graphs.

For approximating the distances, faster approximations for larger constants can be obtained by computing a $k$-spanner, which is a sparse graph that preserves distances up to a multiplicative factor of $k$, and having all nodes learn the entire spanner. Using the \clique spanner constructions of \cite{DBLP:conf/wdag/ParterY18}, this approach gives a $(2k-1)$-approximation for APSP in $\widetilde{O}(n^{1/k})$ rounds, which is still \emph{polynomial} for any constant $k$.

This raises the following fundamental question:
\begin{question}
Can we obtain constant-factor approximations for APSP in \emph{sub-polynomial} time?
\end{question}

If we restrict our attention to SSSP, sub-polynomial $(1+\epsilon)$-approximation is indeed possible~\cite{DBLP:conf/wdag/BeckerKKL17, DBLP:conf/stoc/HenzingerKN16}; in particular, the state of the art is a gradient-descent-based algorithm that obtains a $(1+\epsilon)$-approximation in $O(\epsilon^{-3}$polylog {$n$}) rounds even in the more restricted broadcast version of the \clique model \cite{DBLP:conf/wdag/BeckerKKL17}.  However, these algorithms from prior work provide distances only from a single source.

\subsection{Our contributions} \label{sec:cont}

\paragraph{All-pairs shortest paths.}
As our first main result, we address the above fundamental question by providing the first \emph{polylogarithmic} constant approximations for APSP in the \clique model. Specifically, we show the following.

\begin{theorem}
There is a deterministic $(2+\epsilon)$-approximation algorithm for \emph{unweighted} undirected APSP in the \clique model that takes $O(\frac{\log^2{n}}{\epsilon})$ rounds.
\end{theorem}

We also obtain a nearly $(2+\epsilon)$-approximation in $O(\frac{\log^2{n}}{\epsilon})$ rounds in \emph{weighted} undirected graphs, in the sense that for any distance estimate $d(u,v)$, there is further additive $(1 + \epsilon)w_{uv}$ error in the approximation, where $w_{uv}$ is the weight of the heaviest edge on the shortest $u$-$v$ path.

Our approximation is almost tight for \emph{sub-polynomial} algorithms in the following sense.
As noted by \cite{DBLP:conf/spaa/KorhonenS18}, a $(2-\epsilon)$-approximate APSP in unweighted undirected graphs is essentially equivalent to fast matrix multiplication, so obtaining a better approximation in complexity below $O(n^{0.158})$ would result in a faster algorithm for matrix multiplication in the \clique. Likewise, a sub-polynomial-time algorithm with \emph{any} approximation ratio for directed graphs would give a faster matrix multiplication algorithm \cite{dor2000all}, so our results will likely not extend to directed graphs.

\paragraph{Multi-source shortest paths.}
As our second main result, we show a fast $(1+\epsilon)$-approximation algorithm for the multi-source shortest paths problem (MSSP), which is polylogarithmic as long as the number of sources is $\widetilde{O}(\sqrt{n})$. Specifically, we show the following.

\begin{restatable}{theorem}{mssp}
\label{multi-thm}
There is a deterministic $(1+\epsilon)$-approximation algorithm for the weighted undirected MSSP that takes $$O \bigg( \bigg( \frac{{|S|}^{2/3}}{n^{1/3}} +\log{n} \bigg) \cdot \frac{\log{n}}{\epsilon} \bigg)$$ rounds in the \clique, where $S$ is the set of sources. In particular, the complexity is  $O(\frac{\log^2{n}}{\epsilon})$ as long as $|S| \leq O(\sqrt{n} \cdot (\log{n})^{3/2})$.
\end{restatable}

This is the first sub-polynomial algorithm that obtains such approximations for a set $S$ of polynomial size.
Other advantages of our approach, compared to the previous $(1+\epsilon)$-approximation SSSP algorithm \cite{DBLP:conf/wdag/BeckerKKL17}, is that it is based on simple combinatorial techniques. In addition, our complexity improves upon the complexity of \cite{DBLP:conf/wdag/BeckerKKL17}.

\paragraph{Exact SSSP and diameter approximation.}
In addition to the above, our techniques allow us to obtain a near $3/2$-approximation for the diameter in $O(\frac{\log^2{n}}{\epsilon})$ rounds as well as an  $\widetilde{O}(n^{1/6})$-round algorithm for \emph{exact} weighted SSSP, improving the previous $\widetilde{O}(n^{1/3})$-round algorithm \cite{censor2015algebraic}.
All our algorithms are deterministic.

\subsection{Our techniques}

The main technical tools we develop for our distance computation algorithms are a new \emph{sparse matrix multiplication} algorithm, extending the recent result of \cite{DBLP:conf/opodis/Censor-HillelLT18}, and new deterministic \emph{hopset} construction algorithm for the \clique.

\paragraph{Distance products.} We start from the basic idea of using matrix multiplication to compute distances in graphs. Specifically, if $A$ is the weighted adjacency matrix of a graph $G$, it is well known that distances in $G$ can be computed by iterating the \emph{distance product} $A \star A$, defined as
\[ (A \star A)[i,j] = \min_{k} \bigl( A[i,k] + A[k,j] \bigr)\,,\]
that is, the matrix multiplication over the min-plus semiring.

A simple idea is to apply the recent sparse matrix multiplication algorithm of \cite{DBLP:conf/opodis/Censor-HillelLT18}, with running time that depends on the density of the input matrices. In particular, this allows us to multiply two sparse matrices with $O(n^{3/2})$ non-zero entries in $O(1)$ rounds; note that for the distance product, the zero element is $\infty$. However, using this algorithm for computing distances directly is inefficient, as $A \star A$ can be dense even if $A$ is sparse (e.g.\ a star graph), and hence iterative squaring is not guaranteed to be efficient. Moreover, our goal is to compute distances in general graphs, not only in sparse graphs. Nevertheless, we show that while using \cite{DBLP:conf/opodis/Censor-HillelLT18} directly may not be efficient, we can use sparse matrix multiplication as a basic building block for distance computation in the \clique.

\paragraph{Our distance tools.} The key observation is that many building blocks for distance computation are actually based on computations in sparse graphs or only consider a limited number of nodes. Concrete examples of such tasks include:
\begin{itemize}
    \item \emph{\knearest problem}: Compute distances for each node to the $k$ closest nodes in the graph.
    \item \emph{\sdk problem}: Given a set of sources $S$, compute the distances for each node to the $k$ nearest sources using paths of at most $d$ hops.
    \item \emph{\distthrough problem}: Given a set of nodes $S$ and distances to all nodes in $S$, compute the distances between all nodes using paths through nodes in $S$.
\end{itemize}
For all of these problems, there is a degree of sparsity we can hope to exploit if $k$ or $|S|$ are small enough.
For example, the \sdk problem, requires the multiplication of a dense adjacency matrix and a possibly sparse matrix, depending on the size of $S$. However, for any $S$ of polynomial size the algorithm in \cite{DBLP:conf/opodis/Censor-HillelLT18} is polynomial. An interesting property in this problem, though, is that the output matrix is also sparse. If we look at the \knearest problem, both input matrices are sparse, hence we can use the previous sparse matrix multiplication algorithm. However, this does not exploit the sparsity of this problem to the end: in this problem we are interested only in computing the $k$ nearest nodes to each node, hence there is no need to compute the full output matrix. The challenge in this case is that we do not know the identity of the $k$ closest nodes before the computation. To exploit this sparsity we design new matrix multiplication algorithms, that in particular have the ability to sparsify the matrix throughout the computation, and get a complexity that depends only on the size of the output we are interested in.

\paragraph{Sparse matrix multiplication.}
To compute the above, we design new sparse matrix multiplication algorithms, which differ from~\cite{DBLP:conf/opodis/Censor-HillelLT18} by taking into account also the sparsity of the \emph{output matrix}. For matrix $M$, let $\dens{M}$ denote the \emph{density} of $M$, that is, the average number of non-zero entries on a row.
Specifically, for distance product computation $P = S \star T$, we obtain two variants:
\begin{itemize}
    \item One variant assumes that the sparsity of the output matrix is known.
    \item The other \emph{sparsifies} the output matrix on the fly, keeping only the $\dens{P}$ smallest entries for each row.
\end{itemize}
For these two scenarios, we obtain running times 
\[O\biggl( \frac{( \dens{S}\dens{T}\dens{P} )^{1/3}}{n^{2/3}} + 1 \biggr)\,, \qquad\text{and}\qquad O\biggl( \frac{( \dens{S}\dens{T}\dens{P} )^{1/3}}{n^{2/3}} + \log n \biggr)\,\]
rounds, respectively, improving the running time of the prior sparse matrix multiplication for $\dens{P} = o(n)$. 

This allows us to obtain faster distance tools, by taking into account the sparsity of the output:
\begin{itemize}
    \item We can solve the \knearest problem in $O \left(\left( \frac{k}{n^{2/3} } + \log n \right) \log  k \right)$ rounds.
    \item We can solve the \sdk problem in $O \left(\left( \frac{m^{1/3}|S|^{2/3}}{n} + 1 \right) d \right)$ rounds, where $m$ is the number of edges in the output graph; note that dependence on $d$ becomes linear in order to exploit the sparsity.
\end{itemize}
In concrete terms, with these output-sensitive distance tools we still get subpolynomial running times even when the parameters are polynomial. For example, we can get the distances to the $\widetilde{O}(n^{2/3})$ closest nodes in $\widetilde{O}(1)$ rounds. Note that though our final results are only for undirected graphs, these distance tools work for directed weighted graphs.

\paragraph{Hopsets.}
An issue with our \sdk algorithm is that in order to exploit the sparsity of the matrices, we must perform $d$ multiplications to learn the distances of all the nodes at hop-distance at most $d$ from $S$. Hence, to learn the distances of all the nodes from $S$, we need to do $n$ multiplications, which is no longer efficient.
To overcome this challenge, we use \emph{hopsets}, which are a central building block in many distance computations in the distributed setting \cite{DBLP:conf/stoc/HenzingerKN16, nanongkai2014distributed, DBLP:journals/corr/ElkinN17, elkin2017distributed, elkin2016hopsets}. A $(\beta,\epsilon)$-hopset $H$ is a sparse graph such that the $\beta$-hop distances in $G \cup H$ give $(1+\epsilon)$-approximations of the distances in $G$.
Since it is enough to look only at $\beta$-hop distances in $G \cup H$, using our source detection algorithm together with a hopset allows getting an efficient algorithm for approximating distances, as long as $\beta$ is small enough.

However, the time complexity of all current hopset constructions depends on the size of the hopset \cite{DBLP:journals/corr/ElkinN17, elkin2016hopsets, DBLP:conf/stoc/HenzingerKN16}, in the following way. The complexity of building a hopset of size $n\rho$ is at least $O(\rho)$. This is a major obstacle for efficient shortest paths algorithms, since based on recent existential results there are no hopsets where both $\beta$ and $\rho$ are polylogarithmic \cite{abboud2018hierarchy} (see Section~\ref{section:related-work}.)
Nevertheless, we show that our new distance tools allow to build hopsets in a time that \emph{does not} depend on the hopset size. In particular, we show how to implement a variant of the recent hopset construction of Elkin and Neiman \cite{DBLP:journals/corr/ElkinN17} in $O(\frac{\log^2{n}}{\epsilon})$ rounds. The size of our hopset is $\widetilde{O}(n^{3/2})$, hence constructing it using previous algorithms requires at least $\widetilde{O}(\sqrt{n})$ rounds.

\paragraph{Applying the distance tools.}
As a direct application of our source detection and hopset algorithms, we obtain a multi-source shortest paths (MSSP) algorithm, allowing to compute $(1+\epsilon)$-approximate distances to $\widetilde{O}(\sqrt{n})$ sources in polylogarithmic time. Our MSSP algorithm, in turn, forms the basis of a near $3/2$-approximation for the diameter, and a $(3+\epsilon)$-approximation for weighted APSP. To obtain a $(2+\epsilon)$-approximation for unweighted APSP, the high-level idea is to deal separately with paths that contain a high-degree node and paths with only low-degree nodes. A crucial ingredient in the algorithm is showing that in sparser graphs, we can actually compute distances to a larger set of sources $S$ efficiently, which is useful for obtaining a better approximation. Our exact SSSP algorithm uses our algorithm for finding distances to the $k$-nearest nodes, which allows constructing efficiently the $k$-shortcut graph described in \cite{nanongkai2014distributed, elkin2017distributed}.

\subsection{Additional related work}\label{section:related-work}

\paragraph{Distance computation in the congested clique.} APSP and SSSP are fundamental problems that are studied extensively in various computation models. Apart from the MM-based algorithms in the \clique \cite{censor2015algebraic,le2016further,DBLP:conf/opodis/Censor-HillelLT18}, previous results include also $\widetilde{O}(\sqrt{n})$-round algorithms for exact SSSP and $(2+o(1))$-approximation for APSP \cite{nanongkai2014distributed}. Other distance problems studied in the \clique are construction of hopsets \cite{DBLP:journals/corr/ElkinN17, elkin2016hopsets, DBLP:conf/stoc/HenzingerKN16} and spanners \cite{DBLP:conf/wdag/ParterY18}.

\paragraph{Matrix multiplication in the congested clique.} As shown by \cite{censor2015algebraic}, matrix multiplication can be done in \clique in $O(n^{1/3})$ rounds over semirings, and in $O(n^{1-2/\omega})$ rounds over rings, where $\omega < 2.3728639$ is the exponent of the matrix multiplication~\cite{legall2014powers}. For rectangular matrix multiplication, \cite{le2016further} gives faster algorithms. The first sparse matrix multiplication algorithms for \clique were given by \cite{DBLP:conf/opodis/Censor-HillelLT18}, as discussed above.

\paragraph{Distance computation in the CONGEST model.} The distributed CONGEST model is identical to the \clique model, with the difference that the communication network is identical to the input graph $G$, and nodes can communicate only with their neighbours in each round.
Distance computation is extensively studied in the CONGEST model. The study of  \emph{exact} APSP in weighted graphs has been the focus of many recent papers \cite{DBLP:journals/corr/abs-1811-03337, DBLP:conf/focs/HuangNS17, DBLP:conf/podc/AgarwalRKP18, elkin2017distributed}, culminating in a near tight $\widetilde{O}(n)$ algorithm \cite{DBLP:journals/corr/abs-1811-03337}.
Such results were previously known in unweighted graphs \cite{holzer2012optimal, DBLP:conf/podc/LenzenP13, peleg2012distributed} or for approximation algorithms \cite{lenzen2015fast, nanongkai2014distributed}. Approximate and exact algorithms for SSSP are studied in \cite{DBLP:conf/wdag/BeckerKKL17, nanongkai2014distributed, lenzen2013fast, DBLP:conf/focs/ForsterN18, DBLP:conf/stoc/HenzingerKN16, DBLP:conf/stoc/GhaffariL18, elkin2017distributed}. While near-tight algorithms exist for approximating SSSP \cite{DBLP:conf/wdag/BeckerKKL17}, there is currently a lot of interest in understanding the complexity of \emph{exact} SSSP and \emph{directed} SSSP \cite{DBLP:conf/focs/ForsterN18, DBLP:conf/stoc/GhaffariL18, elkin2017distributed}. The \emph{source detection} problem is studied in \cite{DBLP:conf/podc/LenzenP13}, demonstrating the applicability of this tool for many distance problems such as APSP and diameter approximation in unweighted graphs. An extension for the weighted case is studied in \cite{lenzen2015fast}. Algorithms and lower bounds for approximating the diameter are studied in \cite{DBLP:conf/podc/LenzenP13, le2016further, peleg2012distributed, holzer2012optimal, abboud2016near}.

\paragraph{Distance computation in the sequential setting.} Among the rich line of research in the sequential setting, we focus only on the most related to our work. The pioneering work of Aingworth et al. \cite{aingworth1999fast}, inspired much research on approximate APSP \cite{dor2000all, baswana2010faster, baswana2005all, cohen2001all, berman2007faster} and approximate diameter \cite{roditty2013fast, chechik2014better, DBLP:conf/stoc/BackursRSWW18, berman2007faster}, with the goal of understanding the tradeoffs between the time complexity and approximation ratio. Many of these papers use clustering ideas and hitting set arguments as the basis of their algorithms, and our approximate APSP and diameter algorithms are inspired by such ideas.

\paragraph{Hopsets.} Hopsets are a central building block in distance computation and are studied extensively in various computing models \cite{DBLP:conf/stoc/HenzingerKN16, nanongkai2014distributed, DBLP:journals/corr/ElkinN17, elkin2017distributed, elkin2016hopsets, cohen2000polylog, DBLP:conf/focs/Bernstein09, henzinger2018decremental, shi1999time}. The most related to our work are two recent constructions of Elkin and Neiman \cite{DBLP:journals/corr/ElkinN17}, and Huang and Pettie \cite{huang2019thorup}, which are based on the emulators of Thorup and Zwick \cite{thorup2006spanners}, and are near optimal due to existential results \cite{abboud2018hierarchy}. Specifically, \cite{huang2019thorup} construct $(\beta,\epsilon)$-hopsets of size $O(n^{1+\frac{1}{2^{k+1}-1}})$ with $\beta=O(k/\epsilon)^k$, where recent existential results show that any construction of $(\beta,\epsilon)$-hopsets with worst case size $n^{1+\frac{1}{2^k-1}-\delta}$ must have $\beta = \Omega_k((\frac{1}{\epsilon})^k)$, where $k \geq 1$ is an integer and $\delta > 0$. For a detailed discussion of hopsets see the introduction in \cite{DBLP:journals/corr/ElkinN17, elkin2016hopsets}.

\subsection{Preliminaries}

\paragraph{Notations.} Except when specified otherwise, we assume our graphs are undirected with non-negative integer edge weights at most $O(n^c)$ for a constant $c$.
Given a graph $G=(V,E)$ and $u,v \in V$, we denote by $d_G(u,v)$ the distance between $u$ and $v$ in $G$, and by $d_G^{\beta}(u,v)$ the length of the shortest path of hop-distance at most $\beta$ between $u$ and $v$ in $G$. If $G$ is clear from the context, we use the notation $d(u,v)$ for $d_G(u,v).$

\paragraph{Routing and sorting.} As basic primitives, we use standard routing and sorting algorithms for the \clique model. In the \emph{routing} task, each node holds up to $n$ messages of $O(\log n)$ bits, and we assume that each node is also the recipient of at most $n$ messages.
In the \emph{sorting} task, each node has a list of $n$ entries from an ordered set, and we want to sort these entries so that node $i$ receives the $i$th batch of entries according to the global order of all the input entries. Both of these task can be solved in $O(1)$ rounds~\cite{dolev2012tri,lenzen2013optimal}.

\paragraph{Semirings and matrices.} We assume we are operating over a semiring $(R,+,\cdot, 0, 1)$, where $0$ is the identity element for addition and $1$ is the identity element for multiplication. Note that we do not require the multiplication to be commutative. For the \clique algorithms, we generally assume that the semiring elements can be encoded in messages of $O(\log n)$ bits.

All matrices are assumed to be over the semiring $R$. For convenience, we identify $[n]$ with the node set $V$, and use set $V$ to index the matrix entries.  For matrix $S$, we denote the matrix entry at position $(v,u)$ by $S[v,u]$. For sets $U,W \subseteq V$, we denote by $S[U,W]$ by submatrix obtained by taking rows and columns restricted to $U$ and $W$, respectively.

\paragraph{Hitting sets.} Let $S_v \subseteq V$ be a set of size at least $k$. We say that $A \subseteq V$ is a \emph{hitting set} of $\{S_v\}_{v \in V}$ if in each subset $S_v$ there is a node from $A$. 
We can construct hitting sets easily by adding each node to $A$ with probability $p = \frac{\log{n}}{k}$. This gives a hitting set of expected size $O(\frac{n \log{n}}{k})$, such that w.h.p there is a node from $A$ in each subset $S_v$.
The same parameters are obtained by a recent \emph{deterministic} construction of hitting sets in the \clique \cite{DBLP:conf/wdag/ParterY18}, which gives the following (see Corollary 17 in \cite{DBLP:conf/wdag/ParterY18}).

\begin{lemma} \label{det_hit}
Let $\{S_v \subseteq V\}_{v \in V}$ be a set of subsets of size at least $k$, such that $S_v$ is known to $v$.
There exists a \emph{deterministic} algorithm in the \clique model that constructs a hitting set of size $O(n\log{n}/k)$ in $O((\log{\log{n}})^3)$ rounds.
\end{lemma}

\paragraph{Partitions.}

We will use the following lemmas on partitioning a set of weighted items into equal parts. Note that all the lemmas are constructive, that is, they also imply a deterministic algorithm for constructing the partition in a canonical way.

\begin{lemma}[\cite{DBLP:conf/opodis/Censor-HillelLT18}]\label{lemma:partition-even}
Let $w_1, w_2, \dotsc, w_n$ be natural numbers, and let $W$, $w$ and $k$ be natural numbers such that
$\sum_{i=1}^n w_{i} = W$, $w_{i} \le w$  for all $i$, and
$k$ divides $n$.
Then there is a partition of $[n]$ into $k$ sets $I_1, I_2, \dotsc, I_k$ of size $n/k$ such that
\[\sum_{i \in I_j} w_i \le W/k + w \text{ for all $j$.}\]
\end{lemma}

\begin{lemma}\label{lemma:partition-consecutive}
Let $w_1, w_2, \dotsc, w_n$ be natural numbers, and let $W$, $w$ and $k$ be natural numbers such that
$\sum_{i=1}^n w_{i} = W$, $w_{i} \le w$  for all $i$.
Then there is a partition of $[n]$ into $k$ sets $I_1, I_2, \dotsc, I_k$ such that for each $j$, the set $I_j$ consist of consecutive elements, and
\[\sum_{i \in I_j} w_i \le W/k + w\,.\]
\end{lemma}

\begin{proof}
Construct sets $I_j$ in the partition by starting from the first element, adding new elements to the current set until $\sum_{i \in I_j} w_i \ge W/k$. Since all weights $w_i$ are at most $w$, we have $\sum_{i \in I_j} w_i \le W/k + w$ for all $j$, and since $\sum_{i \in I_j} w_i \ge W/k$, this process generates at most $k$ sets.
\end{proof}

\begin{lemma}\label{lemma:partition-consecutive-2}
Let $w_1, w_2, \dotsc, w_n$ and $u_1, u_2, \dotsc, u_n$ be natural numbers, and let $W$, $U$ $w$, $u$ and $k$ be natural numbers such that
$\sum_{i=1}^n w_{i} = W$ and $\sum_{i=1}^n u_{i} = U$, and we have $w_{i} \le w$ and $u_{i} \le u$ and for all $i$.
Then there is a partition of $[n]$ into $k$ sets $I_1, I_2, \dotsc, I_k$ such that for each $j$, the set $I_j$ consist of consecutive elements, and
\[\sum_{i \in I_j} w_i \le 2(W/k + w)\, \quad \text{and} \quad \sum_{i \in I_j} u_i \le 2(U/k + u)\,. \]
\end{lemma}

\begin{proof}
We begin by applying Lemma~\ref{lemma:partition-consecutive} separately to sequences $w_1, w_2, \dotsc, w_n$ and $u_1, u_2, \dotsc, u_n$. That is, by Lemma~\ref{lemma:partition-consecutive}, there exist the following two partitions of set $[n]$ into sets of consecutive indices:
\begin{itemize}
    \item Partition $J_1, \dotsc, J_k$ such that for any $j$, we have $\sum_{i \in J_j} w_i \le W/k + w$.
    \item Partition $K_1, \dotsc, K_k$ such that for any $j$, we have $\sum_{i \in K_j} u_i \le U/k + u$.
\end{itemize}
Let $i_0 = 0$, and let $i_1, \dotsc, i_{2k}$ be the last elements of sets from partitions $J_1, \dotsc, J_k$ and $K_1, \dotsc, K_k$ in order.

Define sets $I_1, \dotsc, I_k$ as $I_j = \{ i_{2j-2} + 1, \dotsc, i_{2j} \}$; intuitively, this corresponds to taken the \emph{fenceposts} of the partitions $J_1, \dotsc, J_k$ and $K_1, \dotsc, K_k$, and taking every other fencepost to give a new partition into sets of consecutive indices. Clearly $\{ I_j \}$ form a partition of $[n]$, and since each $I_j$ overlaps at most two sets from $J_1, \dotsc, J_k$ and at most two sets from $K_1, \dotsc, K_k$, we have $\sum_{i \in I_j} w_i \le 2(W/k + w)$ and $\sum_{i \in I_j} u_i \le 2(U/k + u)$.
\end{proof}


\section{Matrix multiplication}\label{section:mm}

\subsection{Output-sensitive sparse matrix multiplication}\label{section:sparse-mm}

Our first matrix multiplication result is a output-sensitive variant of sparse matrix multiplication. In the \emph{matrix multiplication problem}, we are given two $n \times n$ matrices $S$ and $T$ over semiring $R$, and the task is to compute the product matrix $P = ST$,
\[ P[u,v] = \sum_{w \in V} S[v,w] T[w,u]\,.\]
Following \cite{censor2015algebraic,le2016further}, we assume for concreteness that in the \clique model, each node receives the row $S[v,V]$ and the column $T[V,v]$ as a local input, and we want to compute the output so that each node locally knows the row $P[v,V]$ of the output matrix.

\paragraph{Matrix densities.} For matrix $S$, we denote by $\nz(S)$ the number of non-zero entries in $S$. Furthermore, we define the \emph{density} $\dens{S}$ as the smallest positive integer satisfying $\nz(S) \le \dens{S} n$.

When discussing the density of a product matrix $P = ST$, we would like to simply consider the density $\dens{P}$; however, for technical reasons, we want to ignore zero entries that appear due to cancellations. Formally, let $\hat{S}$ be the binary matrix defined as
\[ \hat{S}[i,j] =
\begin{cases}
1 & \text{if $S[i,j] \ne 0$},\\
0 & \text{if $S[i,j] = 0$},
\end{cases} \]
and define $\hat{T}$ similarly. Let $\hat{P} = \hat{S}\hat{T}$, where the product is taken over the Boolean semiring. We define the density of the product $ST$, denoted by $\densh{ST}$, as the smallest positive integer satisfying $\nz(\hat{P}) \le \densh{ST} n$. Note that for most of our applications, we operate over semirings where no cancellations can occur in additions, in which case $\densh{ST} = \dens{P}$.

We also note that while we assume that the input matrices are $n \times n$, we can also use this framework for rectangular matrix multiplications, simply by padding the matrices with zeroes to make them square.

\paragraph{Sparse matrix multiplication algorithm.} Our main result for sparse matrix multiplication is the following:

\begin{theorem}\label{theorem:mm}
Matrix multiplication $ST = P$ can be computed deterministically in
\[O\biggl( \frac{( \dens{S} \dens{T} \densh{ST} )^{1/3}}{n^{2/3}} + 1 \biggr)\]
rounds in \clique, assuming we know $\densh{ST}$ beforehand.
\end{theorem}

We note that for all of our applications, the requirement that we know $\densh{ST}$ beforehand is satisfied. However, the algorithm can be modified to work without knowledge of $\densh{ST}$ with the additional cost of multiplicative $O(\log n)$ factor; we simply start with estimate $\densh{ST} = 2$, restarting the algorithm with doubled estimate if the running time for current estimate is exceeded.

The rest of this section gives the proof of Theorem~\ref{theorem:mm}. We start by describing the overall structure of the algorithm, and then detail the different phases of the algorithm individually.

\subsubsection{Algorithm description}

\paragraph{Algorithm parameters.} We define the algorithm parameters $a$, $b$ and $c$ as
\[ a = \frac{(\dens{T}\densh{ST} n)^{1/3}}{\dens{S}^{2/3}} \,, \qquad b = \frac{(\dens{S}\densh{ST} n)^{1/3}}{\dens{T}^{2/3}} \,, \qquad c = \frac{(\dens{S}\dens{T} n)^{1/3}}{\densh{ST}^{2/3}}\,.\]
These parameter will control how we split the matrix multiplication task into independent subtasks. To see why these parameters are chosen in this particular way, we note that we will require that $abc = n$, and the final running time of the algorithm will be $O\bigl(\dens{S} a/n + \dens{T} b/n + \densh{ST} c/n + 1\bigr)$ rounds,
which is optimized by the selecting $a$, $b$ and $c$ as above; this gives the running time in Theorem~\ref{theorem:mm}.

For simplicity, we assume that $a$, $b$ and $c$ are integers. If not, taking the smallest greater integer will cause at most constant overhead.

\paragraph{Algorithm overview.} Our algorithm follows the basic idea of the classical 3D matrix multiplication algorithm, as presented for \clique by \cite{censor2015algebraic}, and as adapted for the sparse matrix multiplication by~\cite{DBLP:conf/opodis/Censor-HillelLT18}. That is, we want to reduce the matrix multiplication task into $n$ smaller instances of matrix multiplication, and use a single node for each one of these. However, due to the fact that we are working with sparse matrices, we have to make certain considerations in our algorithm:
\begin{itemize}
    \item Whereas the 3D matrix multiplication splits the original multiplication into $n$ products of $n^{2/3} \times n^{2/3}$ matrices, we would ideally like to split into $n$ products of shape $(n/b \times n/c) \times (n/c \times n/a)$.
    \item Unlike in the dense case, we also have to make sure that all of our $n$ products are equally sparse. While this could be achieved using randomization similarly to the triangle detection algorithms of \cite{Pandurangan2018,expander2018}, we want to do this deterministically.
\end{itemize}

With above considerations in mind, we now give an overview of our sparse matrix multiplication algorithm:

\begin{oframed}
\begin{enumerate}
\item We compute a partition of the matrix multiplication task $P = ST$ into $n$ subtasks $P^v = S^v T^v$, where $v \in V$, so that each $P^v$ is a $n/b \times n/a$ matrix, $S^v$ and $T^v$ are submatrices of $S$ and $T$, respectively, and we have
\begin{align*}
  \nz(S^v) & = O(\dens{S} n/bc) =  O(\dens{S} a) \qquad \text{and} \\
  \nz(T^v) & = O(\dens{T} n/ac) = O(\dens{T} b)
\end{align*}
for all $v \in V$. This step takes $O(1)$ rounds. (Section~\ref{section:mm-cube}.)
\item Each node $v$ learns the matrices $S^v$ and $T^v$, and computes the product $P^v = S^v T^v$. Note that after this step, some of the matrices $P^v$ may be very dense. This step takes $O\bigl(\dens{S} a/n + \dens{T} b/n + 1\bigr)$ rounds. (Section~\ref{section:mm-products1}.)
\item We balance the output matrices $P^v$ so that each node holds $O(\densh{ST} n/ab) = O(\densh{ST} c)$ values that need to be summed to obtain the final output matrix $P$. This is achieved by duplicating those subtasks where the output is too dense. This step takes $O\bigl(\dens{S} a/n + \dens{T} b/n + 1\bigr)$ rounds. (Section~\ref{section:mm-products2}.)
\item The intermediate values obtained in Step~4 are summed together to obtain the output matrix $P$. This step takes $O\bigl(\densh{ST} c/n + 1\bigr)$ rounds. (Section~\ref{section:mm-summation}.)
\end{enumerate}
\end{oframed}

Note that the total running time of the above algorithm will be $O\bigl(\dens{S} a/n + \dens{T} b/n + \densh{ST} c/n + 1\bigr)$
rounds, which by the choice of $a$, $b$ and $c$ is as stated in Theorem~\ref{theorem:mm}. Note that Steps (1) and (2) are essentially streamlined versions of corresponding tools from \cite{DBLP:conf/opodis/Censor-HillelLT18}, while Steps (3) and (4) are new.

\subsubsection{Cube partitioning}\label{section:mm-cube}

We say that a \emph{subcube} of $V^3$ is a set of form $V_1 \times V_2 \times V_3$, where $V_1, V_2, V_3 \subseteq V$. Note that such a subcube corresponds to a matrix multiplication task $S[V_1, V_2] T[V_2, V_3]$.
Thus, a partition of the cube $V^3$ into subcubes corresponds to a partition of the original matrix multiplication into smaller matrix multiplication tasks, as discussed in the overview; see also Figure~\ref{fig:mm-1}.

\begin{lemma}\label{lemma:cube-partitioning}
There is a \clique algorithm running in $O(1)$ rounds that produces globally known a partition of $V^3$ into $n$ disjoint subcubes $V_{i}$ such that for each subcube $V_{i} = V_i^{S} \times V_i^{T} \times V_i^{P}$, we have $|V^S_i| = O(n/b)$, $|V^T_i| = O(n/a)$ and the total number of non-zero entries is
\begin{enumerate}
    \item $O\bigl(\dens{S} a + n \bigr)$ in the submatrix $S[V_i^S, V_i^P]$, and
    \item $O\bigl(\dens{T} b + n \bigr)$ in the submatrix $T[V_i^P, V_i^T]$.
\end{enumerate}
\end{lemma}

\begin{proof}{} We start by partitioning the input matrices into equally sparse `slices'. Specifically, we do the following:
\begin{enumerate}
    \item All nodes $v$ broadcast the number of non-zero entries on row $v$ of $S$. Based on this information, all nodes deterministically compute the same partition of $V$ into $b$ sets $C^S_1, C^S_2, \dotsc, C^S_b$ of size $O(n/b)$ such that $\nz(S[C^S_i, V]) = O(\dens{S}n/b + n)$; such partition exists by Lemma~\ref{lemma:partition-even}.
    \item Using the same procedure as above, the nodes compute a partition of $V$ into $a$ sets $C^T_1, C^T_2, \dotsc, C^T_a$ of size $O(n/a)$ such that $\nz(T[V, C^T_i]) = O(\dens{T}n/a + n)$.
\end{enumerate}
There are now $ab$ pairs $(C^S_i, C^T_j)$, each corresponding to a subcube $C^S_i \times V \times C^T_j$. We now partition each of these subcubes in parallel as follows. First, we partition the nodes into $ab$ sets of $n/ab = c$ nodes, each such set $B_{ij}$ corresponding to a pair of indices $(i,j) \in [b] \times [a]$. The final partition is now computed as follows:
\begin{enumerate}
    \item Nodes redistribute the input matrices so that node $v$ holds column $v$ of $S$ and row $v$ of $T$. This can be done in $O(1)$ rounds.
    \item In parallel for each $i$ and $j$, node $v$ sends the number of non-zero elements in $S[C^S_i,v]$ and $T[v, C^T_j]$ to all nodes in $B_{ij}$.
    \item Nodes in $B_{ij}$ compute a partition of $V$ into $c$ sets $C^{ij}_1, C^{ij}_2, \dotsc, C^{ij}_c$ of \emph{consecutive} indices such that the number of non-zero entries in $S[C^{S}_i, C^{ij}_k]$ is $O\bigl(\dens{S} n/bc + n \bigr)$ and the number of non-zero entries in $T[C^{ij}_k, C_j^T]$ is $O\bigl(\dens{T} n/ac + n \bigr)$; such partition exists by Lemma~\ref{lemma:partition-consecutive-2}.
    \item For each $C^{ij}_k$, the $k$th node in $B_{ij}$ broadcasts the first and last index of $C^{ij}_k$ to all other nodes, allowing all nodes to reconstruct these sets.
\end{enumerate}
The subcubes $C^S_i \times C^{ij}_k \times C^T_j$ are now known globally. Furthermore, since $n/bc = a$ and $n/ac = b$, the subcubes satisfy the requirements of the claim by construction.
\end{proof}

\begin{figure}
\center
\includegraphics[scale=0.7]{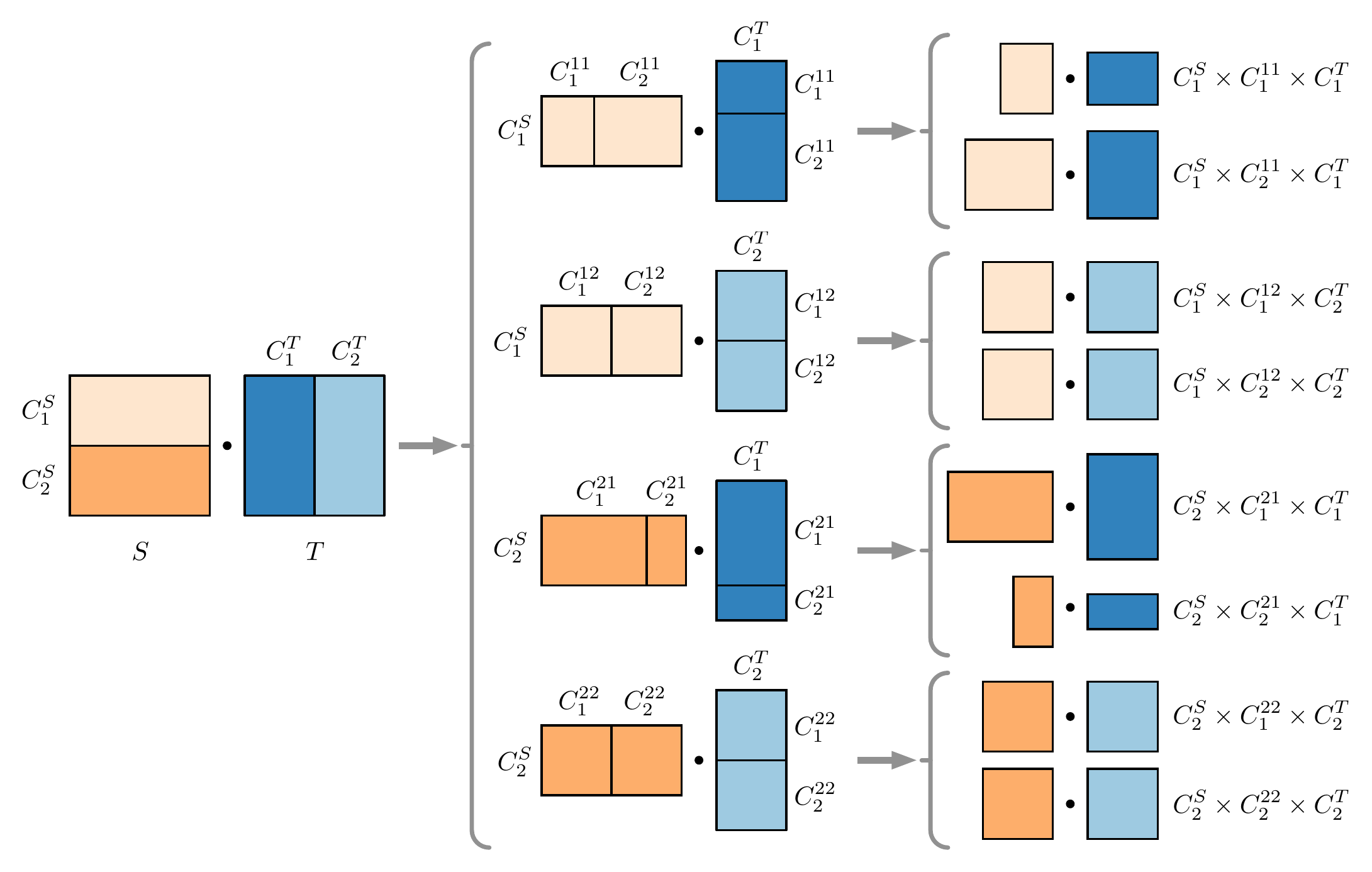}    
\caption{Matrix multiplication decomposed into subtasks using Lemma~\ref{lemma:cube-partitioning} for $a = b = c = 2$ and $n = 8$. The notation follows the proof of Lemma~\ref{lemma:cube-partitioning}.}\label{fig:mm-1}
\end{figure}

\subsubsection{Intermediate products}\label{section:mm-products1}

\paragraph{Balancing.} As an auxiliary tool, we want to solve a balancing task where each node $i$ has $n$ weighted entries with weights $w_{i1}, w_{i2}, \dotsc, w_{in} \in \mathbb{N}$ that satisfy
\[ \sum_{i,j} w_{ij} \le W\,, \qquad w_{ij} \le n\,,\]
and the task is to re-distribute the entries so that each node has $n$ entries with total weight $O(W/n + n)$. Concretely, we assume each weighted entry consists of the weight and $O(\log n)$ bits of data.

\begin{lemma}\label{lemma:balancing}
The above balancing task can be solved in $O(1)$ rounds in \clique.
\end{lemma}

\begin{proof}{}
As a first step, we globally learn the distribution of different weights, and compute a globally known ordering for the weighted entries:
\begin{enumerate}
    \item All nodes send the number of entries with weight exactly $i$ to node $i$, with node $1$ handling entries with both weight $0$ and $1$. Node $i$ broadcasts the total to all other nodes.
    \item Nodes sort the weighted entries using Lenzen's sorting algorithm~\cite{lenzen2013optimal}.
\end{enumerate}
Since all nodes know the distribution of the weights, all nodes can now locally compute what entries other nodes hold. This gives us a globally consistent numbering for the weighted entries, which we can use to solve the balancing task:
\begin{enumerate}
    \item Nodes locally compute a partition of weighted entries into $n$ sets of size $n$ with total weight $O(W/n + n)$; such partition exists by Lemma~\ref{lemma:partition-even}.
    \item Each set is assigned to a separate node. Nodes redistribute the entries so that each node receives their assigned set in $O(1)$ rounds.
\end{enumerate}
All steps clearly take $O(1)$ rounds.
\end{proof}

\paragraph{Computing intermediate products.} We now show how to compute the intermediate products given by the cube partition of Lemma~\ref{lemma:cube-partitioning}. The following lemma is in fact more general; we will use it as a subroutine in subsequent steps of the algorithm.

\begin{lemma}\label{lemma:product-algorithm}
Let $V_1, V_2, \dotsc, V_n$ be a partition of $V^3$ as in Lemma~\ref{lemma:cube-partitioning}, and let $\sigma \colon V \to [n]$ be a (not necessarily bijective) function that is known to all nodes. There is a \clique algorithm running in $O(\dens{S} a/n + \dens{T} b/n + 1)$ rounds such that after the completion of the algorithm, each node $v$ locally knows the product
\[ P^{\sigma(v)} = S[V_{\sigma(v)}^S, V_{\sigma(v)}^P]T[V_{\sigma(v)}^P, V_{\sigma(v)}^T]\,.\]
\end{lemma}

\begin{proof}{}
For each $i, j \in V$, define $w_{ij}$ as
\[
w_{ij} =
\begin{cases}
|\{ v \in V \colon (i,j) \in V_{\sigma(v)}^S \times V_{\sigma(v)}^P \}|, & \text{ if $S[i,j]$ is non-zero, and}\\
0, & \text{ otherwise,}
\end{cases}
\]
that is, $w_{ij}$ is the number of times $S[i,j]$ appears in matrices $S[V_{\sigma(v)}^S, V_{\sigma(v)}^P]$, or $0$ if $S[i,j]$ is a zero entry. Clearly we have $w_{ij} \le n$, and since the partition satisfies the conditions of Lemma~\ref{lemma:cube-partitioning}, we have that
\[ \sum_{i,j} w_{ij} = \sum_{v \in V}\nz\bigl(S[V_{\sigma(v)}^S, V_{\sigma(v)}^P]\bigr) = O\bigl(\dens{S} na + n^2 \bigr)\,.\]
All nodes can compute the values $w_{ij}$ locally for their own row, since they depend only on the partition and function $\sigma$.

We now distribute the entries of input matrix $S$ so that each node learns the matrix $S[V_{\sigma(v)}^S, V_{\sigma(v)}^P]$:
\begin{enumerate}
    \item Using Lemma~\ref{lemma:balancing}, we balance the input entries so that each node holds $n$ entries with total weight $O(\dens{S} a + n)$. Specifically, for each entry a node receives, it receives the value $S[i,j]$ along with the index $(i,j)$.
    \item Since the nodes know the partition and function $\sigma$, each node computes to which nodes it needs to send each of the entries it received in the first step. Since each entry is duplicated $w_{ij}$ times, each node needs to send $O(\dens{S} a + n)$ messages, and dually, each node needs to receive a submatrix of $S$ with $O(\dens{S} a + n)$ entries. These messages can be delivered in $O(\dens{S}a/n + 1)$ rounds.
\end{enumerate}
By identical argument, each node $v$ can learn the matrix $T[V_{\sigma(v)}^P, V_{\sigma(v)}^T]$ in $O(\dens{T} b/n + 1)$ rounds and compute $P^{\sigma(v)}$ locally.
\end{proof}

\subsubsection{Balanced intermediate products}\label{section:mm-products2}

We say that an \emph{intermediate value} in the matrix multiplication is a value
\[p_{vWu} = S[v,W] T[W,u] = \sum_{w \in W} S[v,w]T[w,u]\,.\]
That is, an intermediate value is a partial sum of products for a single position of the output matrix. For concreteness, we encode these in the \clique implementation as tuples $(p_{vWu}, v, u)$.

\begin{lemma}\label{lemma:product-balancing}
There is a \clique algorithm running in $O(\dens{S} a/n + \dens{T} b/n + 1)$ rounds such that after the completion of the algorithm,
\begin{enumerate}
    \item each node holds $O(\densh{ST} n/ab) = O(\densh{ST} c)$ non-zero intermediate values, and
    \item each non-zero elementary product $S[v,w]T[w,u]$ in the matrix multiplication is included in exactly one intermediate value held by some node.
\end{enumerate}
\end{lemma}

\begin{proof}{} As the first part of the algorithm, we compute all the intermediate product matrices and learn their densities:
\begin{enumerate}
    \item Compute a partition of $V^3$ using Lemma~\ref{lemma:cube-partitioning}.
    \item Applying Lemma~\ref{lemma:product-algorithm} with $\sigma_1(v) = v$, each node computes the matrix
    \[ P^{v} = S[V_{v}^S, V_{v}^P]T[V_{v}^P, V_{v}^T]\,.\]
    This takes $O(\dens{S} a/n + \dens{T} b/n + 1)$ rounds.
    \item Each node $v$ broadcasts the number of non-zero entries in $P^v$ to all other nodes.
\end{enumerate}
Next, we want to balance the dense intermediate product matrices between multiple nodes. We cannot do this directly, but we can instead duplicate the products:
\begin{enumerate}
    \item Construct a function $\sigma_2$ so that for each $v$ with $\nz(P^v) \ge \densh{ST} c$, there are at least $\lfloor \frac{\nz(P^v)}{\densh{ST} c}\rfloor$ values $u \in V$ satisfying $\sigma_2(u) = v$. To see that this is possible, we observe that by the definition of $\densh{ST}$, there are at most $\densh{ST} n$ positions where matrices $P^v$ can have non-zero entries and each such position is duplicated $c$ times in the partition of the cube $V^3$, implying that $\sum_{v \in V} \nz(P^v) \le \densh{ST} nc$. Thus, we have
    \[ \sum_{v \in V} \lfloor \frac{\nz(P^v)}{\densh{ST} c}\rfloor \le \sum_{v \in V} \frac{\nz(P^v)}{\densh{ST} c} = \frac{1}{\densh{ST} c} \sum_{v \in V} \nz(P^v) \le \frac{\densh{ST} n c}{\densh{ST} c} = n\,.\]
    This step can be done locally using information obtained in the first part of the algorithm.
    \item Apply Lemma~\ref{lemma:product-algorithm} with $\sigma_2$. This takes $O(\dens{S} a/n + \dens{T} b/n + 1)$ rounds.
    \item For each $u$, each node $v$ with $\sigma_1(v) = u$ or $\sigma_2(v) = u$ assumes responsibility for $O(\densh{ST} c)$ entries of the matrix $P^u$ and discards the rest. More specifically, node $v$ determines $i$ such that $v$ is the $i$th node responsible for $P^u$, splits the non-zero entries of $P^u$ into $\lceil\frac{\nz(P^u)}{\densh{ST} c}\rceil$ parts and selects the $i$th part; if both $\sigma_1(v) = u$ and $\sigma_2(v) = u$, then $v$ selects two parts. This step can be done locally based on information obtained earlier.
\end{enumerate}
After the completion of the algorithm, each node has $O(\densh{ST} c)$ intermediate values from at most two matrices $P^v$. The total running time is $O(\dens{S} a/n + \dens{T} b/n + 1)$ rounds.
\end{proof}

\subsubsection{Balanced summation}\label{section:mm-summation}

\begin{lemma}\label{lemma:summation-main}
Assume that the non-zero intermediate values of the matrix multiplication $P = ST$ have been computed as in Lemma~\ref{lemma:product-balancing}. Then there is a \clique algorithm running in $O(\densh{ST}c/n + 1)$ rounds that computes the output matrix $P$.
\end{lemma}

\begin{proof}
We start by initializing each row of the output matrix to all zero values. Our objective is to accumulate the intermediate values to this initial matrix, with each node $v$ being responsible for the row $v$ of the output matrix.

All nodes split their intermediate values into $O(\densh{ST}c/n)$ sets of at most $n$ intermediate values. We then repeat the algorithm $O(\densh{ST} c / n + 1)$ times, each repetition accumulating one set of $n$ intermediate values from each node:
\begin{enumerate}
    \item Nodes sort the $n^2$ intermediate values being processes globally by matrix position. This takes $O(1)$ rounds using Lenzen's sorting algorithm.
    \item Each node locally sums all intermediate products it received corresponding to the same position.
    \item Each node broadcasts the minimum and maximum matrix position it currently holds. Nodes use this information to deduce if the same matrix position occurs in multiple nodes. If so, all sums corresponding to that position are sent to the smallest id node having that position; each node sends at most one sum and receives at most one sum from each other node, so this step takes $O(1)$ rounds.
    \item If a node received new values from other nodes, these are now added to the appropriate sum.
    \item All nodes now hold sums corresponding to at most $n$ matrix positions. Using Lenzen's routing algorithm, we redistribute these so that node $v$ obtains sums corresponding to positions on row $v$; this takes $O(1)$ rounds. Node $v$ then accumulates these sums to the output matrix.
\end{enumerate}
Clearly, after completion of all repeats, we have obtained the output matrix $P$. Since each repeat of the accumulation algorithm takes $O(1)$ rounds, and there are $O(\densh{ST}c/n + 1)$ repeats, the whole process takes $O(\densh{ST} c / n + 1)$ rounds.
\end{proof}

\subsection{Matrix multiplication with sparsification}\label{section:mm-filtered} 

Our second matrix multiplication result is a variant of sparse matrix multiplication where we control the density of the output matrix.

\paragraph{Problem definition.}  In this section, we assume our semiring $(R,+,\cdot, 0, 1)$ satisfies the following conditions:
\begin{enumerate}
    \item there is a total order $<$ on $R$, and
    \item the addition operation $+$ satisfies $x + y = \min(x,y)$, where $\min$ is taken in terms of order $<$.
\end{enumerate}

Let $P$ be a matrix and let $0 \le \denssymbol \le n$ be an integer. We define the \emph{$\denssymbol$-filtered} version of $P$ as a matrix $\s{P}$ such that each row of $\s{P}$ contains $\denssymbol$ smallest entries of $P$, that is,
\begin{enumerate}
    \item either $\s{P}[v,u] = 0$ or $\s{P}[v,u] = P[v,u]$,
    \item if row $v$ of $P$ has $\sigma$ non-zero entries, then row $v$ for $\s{P}$ has $\min(\sigma,\denssymbol)$ non-zero entries, and
    \item if $\s{P}[v,u] = 0$ and $P[v,u] \ne 0$, then $\max_w \s{P}[v,w] \le P[v,u]$.
\end{enumerate}
When $\denssymbol$ is clear from the context, we call $\s{P}$ a filtered version of $P$. 

Let $S$ and $T$ be the input matrices over the semiring $R$ and denote $P = ST$. In the \emph{filtered matrix multiplication problem}, we are given the input matrices $S$ and $T$ along with an output density parameter $\denssymbol$, and the task is to compute a $\denssymbol$-filtered output matrix $\s{P}$.

\paragraph{Filtered matrix multiplication algorithm.}
We now prove an analogue of Theorem~\ref{theorem:mm} for the filtered matrix multiplication problem. Our result is the following:

\begin{restatable}{theorem}{mmfiltered}\label{theorem:mm-filtered}
Assume we know beforehand a set $R' \subseteq R$ of semiring elements that can appear during the computation of the product $ST$, and $\size{R'} = W$.
Then $\rho$-filtered matrix multiplication can be computed in
\[O\biggl( \frac{( \dens{S} \dens{T} \denssymbol )^{1/3}}{n^{2/3}} + \log W \biggr)\]
rounds in the \clique. In particular, when the underlying semiring is the min-plus semiring and the matrix entries are integers of absolute value at most $O(n^c)$ for a constant $c \ge 1$, the running time is
\[O\biggl( \frac{( \dens{S} \dens{T} \denssymbol )^{1/3}}{n^{2/3}} + \log n \biggr)\]
rounds.
\end{restatable}

The high-level proof idea is largely the same as for Theorem~\ref{theorem:mm}. However, since we now cannot guarantee that the output matrix, and by extension the results of the subtasks, are sparse, we will have to perform part of the filtering before we sum the intermediate results together -- in essence, between Steps~(2) and (3) of the sparse matrix multiplication algorithm.

For this first filtering step, we group the nodes into sets responsible for subsets of intermediate results on the same row of the output matrix, and perform multiple parallel binary searches to identify which intermediate results can be safely filtered out. This allows us to discard all but $O(\denssymbol n c)$ intermediate entries, which gives us the additive $O(\log W)$ overhead in the algorithm.  After filtering, we follow the same strategy as in the sparse matrix multiplication algorithm -- namely, Steps~(3) and (4) -- to sum together the intermediate entries. Each node will then hold a \emph{partially filtered} row of the output matrix $P = ST$. We can then perform a second filtering step locally to obtain $\s{P}$.

\subsubsection{Algorithm description}

\paragraph{Cube partitioning in detail.}

\begin{figure}
\center
\includegraphics[scale=0.7]{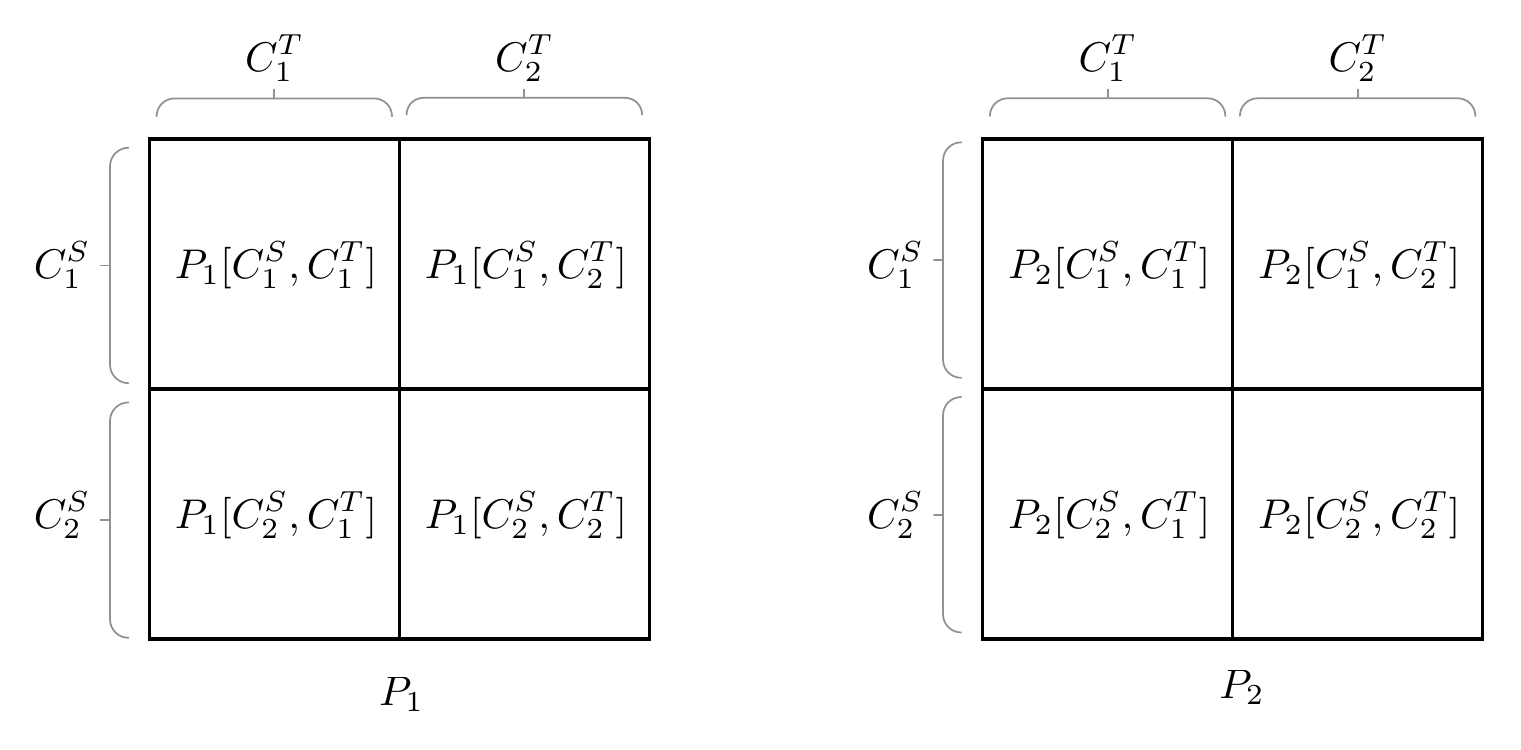}    
\caption{Example of combining the intermediate product matrices to form the matrices $P_k$ for $a = b = c = 2$ and $n = 8$.}\label{fig:mm-2}
\end{figure}

For the filtered version of the matrix multiplication, we need to keep more detailed track of the partitioning given by Lemma~\ref{lemma:cube-partitioning}. Specifically, Lemma~\ref{lemma:cube-partitioning} gives us a partitioning of the cube $V^3$ into $n$ subcubes of form
\[ C_i^S \times C^{ij}_k \times C^T_j\,,\qquad (i,j,k) \in [b] \times [a] \times [c] \,,\]
which defines the matrix multiplication sub-tasks $P^v = S^v T^v$. In this section, we write the products $S^v T^v$ instead as $S[C_i^S, C^{ij}_k]T[C^{ij}_k, C^{T}_j]$.

\paragraph{Filtering.}

To precisely formulate our filtering steps, we define $n \times n$ matrices $P_k$ for $k \in [c]$ by setting
\[ P_k[C_i^S, C^{T}_j] =  S[C_i^S, C^{ij}_k]T[C^{ij}_k, C^{T}_j]\,.\]
That is, each $P_k$ is obtained by combining the smaller $n/b \times n/a$ matrices $S[C_i^S, C^{ij}_k]T[C^{ij}_k, C^{T}_j]$ into an $n \times n$ square matrix; see Figure~\ref{fig:mm-2}. In particular, the final unfiltered product matrix $P = ST$ is obtained by summing the $P_k$ matrices as $P = \sum_{k = 1}^c P_k$.

We can obtain the filtered version $\s{P}$ of the final output matrix $P$ by first applying the filtering to the matrices $P_k$ and then filtering the sum of those matrices again, that is,
\[ Q = \sum_{k = 1}^c \s{P_k}\,, \hspace{20mm} \s{P} = \s{Q}\,.\]
Moreover, by the definition of filtering, each matrix $\s{P_k}$ has density at most $\denssymbol$, so the total number of entries in matrices $\s{P_k}$ is $O(\denssymbol n c)$. The high-level idea of our algorithm is now that after computation of the matrix multiplication subtasks $S[C_i^S, C^{ij}_k]T[C^{ij}_k, C^{T}_j]$, we identify which intermediate values are entries of the matrices $\s{P_k}$ and discard the rest, which allows us to complete the summation step within the desired time budget.

\paragraph{Algorithm overview.} More precisely, the filtered matrix multiplication algorithm proceeds in following steps:

\begin{oframed}
\begin{enumerate}
\item We compute a partition of the matrix multiplication task $P = ST$ into $n$ sparse subtasks
\[ P_k[C_i^S, C^{T}_j] =  S[C_i^S, C^{ij}_k]T[C^{ij}_k, C^{T}_j]\,,\]
where $S[C_i^S, C^{ij}_k]$ and $T[C^{ij}_k, C^{T}_j]$ are submatrices of $S$ and $T$, respectively, and we have
\begin{align*}
  \nz(S[C_i^S, C^{ij}_k]) & = O(\dens{S} a) \qquad \text{and} \\
  \nz(T[C^{ij}_k, C^{T}_j]) & = O(\dens{T} b)\,.
\end{align*}
This step takes $O(1)$ rounds. (Identical to Section~\ref{section:sparse-mm}.)
\item Each node $v$ learns the matrices $S[C_i^S, C^{ij}_k]$ and $T[C^{ij}_k, C^{T}_j]$ for a single $(i,j,k) \in [b] \times [a] \times [c]$, and computes their product $P_k[C_i^S, C^{T}_j]$. This step takes $O\bigl(\dens{S} a/n + \dens{T} b/n + 1\bigr)$ rounds. (Identical to Section~\ref{section:sparse-mm}.)
\item For each $k \in [c]$ and row $\ell$ of $P_k$, the nodes that hold entries from row $\ell$ of $P_k$ perform a distributed binary search to find the $\rho$th largest value on row $\ell$. This allows us to identify which intermediate entries that will appear in matrices $\s{P_k}$ and to discard the rest. This step takes $O(\log W)$ rounds. (Lemma~\ref{lemma:filtering-bs}.)
\item We balance the entries of matrices $\s{P_k}$ so that each node holds $O(\denssymbol n/ab) = O(\denssymbol c)$ values that need to be summed to obtain the matrix $Q = \sum_{k} \s{P_k}$. This is achieved by duplicating those subtasks where the output contains too many entries from matrices $\s{P_k}$. This step takes $O\bigl(\dens{S} a/n + \dens{T} b/n + 1\bigr)$ rounds. (Lemma~\ref{lemma:filtering-balancing}.)
\item The intermediate values obtained in Step~5 are summed together to obtain the matrix~$Q$ so that each node holds a single row of $Q$. This step takes $O\bigl(\denssymbol c/n + 1\bigr)$ rounds. (Identical to Section~\ref{section:sparse-mm}.)
\item Each node discards all but the $\denssymbol$ smallest entries on their row of $Q$ to obtain final output $\s{P}$. This step requires only local computation.
\end{enumerate}
\end{oframed}

\subsubsection{Filtering}

\paragraph{Node grouping for filtering.} We identify in an arbitrary fashion each node $v \in V$ with a triple $(i,j,k) \in [b] \times [a] \times [c]$. This node $v$ is responsible for computing the intermediate product
\[ S[C_i^S, C^{ij}_k]T[C^{ij}_k, C^{T}_j]\,.\]
Furthermore, for $i \in [b]$ and $k \in [c]$, let $B_{ik} \subseteq V$ be the set of nodes corresponding to triples $(i,j,k)$ for $j \in [a]$. In particular, for fixed $i$ and $k$, the nodes in $B_{ik}$ are responsible for computing the products that give the rows $\ell \in C^S_i$ of the matrix $P_k$. Finally, we observe that $\size{B_{ik}} = a$.

\paragraph{Filtering intermediate products.} Recall that we assume that we know a set $R' \subseteq R$ of size $W$ of semiring elements that can appear during the computation of the product $ST$. We now show how the nodes compute the \emph{cutoff value} for each row of each matrix $\s{P_k}$. Formally, to define the cutoff value for row $\ell$ of matrix $P_k$, consider the set
$\bigl\{ (P_k[\ell, i], i) \colon P_k[\ell, i] \ne 0 \bigr\}$
equipped with the natural ordering, that is, $(r,s) < (r',s')$ if $r < r$ or if $r = r'$ and $s < s'$. The cutoff value is the $\denssymbol$th largest element in the set, or the largest if the set has size less than $\denssymbol$.

\begin{lemma}\label{lemma:filtering-bs}
Assume that each node has computed the corresponding product $S[C_i^S, C^{ij}_k]T[C^{ij}_k, C^{T}_j]$. Then there is a \clique algorithm running in $O(\log W)$ rounds such that after the completion of the algorithm, each node $v \in B_{ik}$ knows the cutoff value for row $\ell$ in $P_k$ for all rows $\ell \in C^S_i$.
\end{lemma}

\begin{proof}
Fix $i \in [b]$, $k \in [c]$ and a row $\ell \in C_i^S$. We use binary search to find the smallest value $r$ such that there are at least $\denssymbol$ non-zero entries at most $r$ on row $\ell$ of $P_k$. We fix a node $u \in B_{ik}$ the \emph{coordinator} for row $\ell$ of $P_k$. The binary search now proceeds as follows:
\begin{enumerate}
    \item In the first round of the search, all nodes $v \in B_{ik}$ send the coordinator $u$ the number of non-zero entries they have on row $\ell$. If the total number of non-zero entries is at most $\denssymbol$, we are done. Otherwise, the coordinator sets $r_1 = \min R'$ and $r_2 = \max R'$, and we proceed with the binary search.
    \item In the subsequent rounds, the coordinator selects the value $r' \in R'$ halfway between $r_1$ and $r_2$ and broadcasts it to all nodes $v \in B_{ik}$. The nodes $v \in B_{ik}$ then send the coordinator the number of values at most $r'$ they have on row~$\ell$. If the total number is less than $\denssymbol$, the coordinator sets $r_1 = r'$, and otherwise the coordinator sets $r_2 = r'$. We repeat this step until there are no values in $R'$ between $r_1$ and $r_2$.
\end{enumerate}
The binary search finishes $O(\log W)$ iterations. Clearly $r = r_2$ is the desired value. To obtain the final cutoff value, all nodes $v \in B_{ik}$ tell the coordinator (1) how many values strictly less than $r$ they have, and (2) how many values equal to $r$ they have. The coordinator then determines how many values equal to $r$ should be kept, and determines which node in $B_{ik}$ holds the entry corresponding to the cutoff value. The coordinator then queries that node to obtain the final cutoff value $(r,s)$, and broadcasts it to $B_{ik}$. The post-processing after the binary search takes $O(1)$ rounds.

It remains to show that we can execute all the binary searches in parallel. For each set $B_{ik}$ of $a$ nodes, the nodes have to execute $O(n/b)$ binary searches, and we assign each node $v \in B_{ik}$ as a leader for $O(n/ab)$ binary searches. Thus, for each iteration of the binary search, each node $v \in B_{ik}$ needs to send and receive $O(an/ab) = O(n/b) = O(n)$ messages in the coordinator role. Dually, each node participates in $O(n/b)$ binary searches, so it needs to send and receive $O(n/b) = O(n)$ messages for each iteration of the binary search in the participant role. Thus, one iteration of all binary searches can be completed in $O(1)$ rounds.
\end{proof}

Given the cutoff values computed by Lemma~\ref{lemma:filtering-bs}, all nodes can filter their local products to discard all entries that don't appear in matrices $\s{P_k}$. However, some nodes may still hold too many entries from matrices $\s{P_k}$, so we will balance the entries in a similar way as in Lemma~\ref{lemma:product-balancing}:

\begin{lemma}\label{lemma:filtering-balancing} Assume that the cutoff values of rows of matrices $P_k$ are computed as per Lemma~\ref{lemma:filtering-bs}.
Then there is a \clique algorithm running in $O(\dens{S} a/n + \dens{T} b/n + 1)$ rounds such that after the completion of the algorithm,
\begin{enumerate}
    \item each node holds $O(\denssymbol c)$ non-zero entries from matrices $\s{P_k}$, and
    \item each non-zero entry of matrices $\s{P_k}$ held by exactly one node.
\end{enumerate}
\end{lemma}

\begin{proof} For node $v \in V$, let $w_v$ be the number of entries from the matrices $\s{P_k}$ the node $v$ holds after the product $S[C_i^S, C^{ij}_k]T[C^{ij}_k, C^{T}_j]$ is filtered using Lemma~\ref{lemma:filtering-bs}. Let $\alpha_i = \lceil \frac{\size{C^S_i}}{n/b} \rceil$; note that $\alpha_i = O(1)$. We now want to apply Lemma~\ref{lemma:product-algorithm}; to this end, we construct a helper assignment function $\sigma$ so that
\begin{itemize}
    \item for each $v$ with $w_v \ge \denssymbol \alpha_i c$, there are at least $\lfloor \frac{w_v}{\denssymbol \alpha_i c}\rfloor$ values $u \in V$ satisfying $\sigma(u) = v$, and
    \item for each $v \in B_{ik}$, all $u$ such that $\sigma(u) = v$ also satisfy $u \in B_{ik}$.
\end{itemize}
To see that this is possible, consider a set $B_{ik}$; since these nodes together hold $\size{C^S_i}$ rows of the matrix $\s{P_k}$, we have that  $\sum_{v \in B_{ik}} w_v \le \denssymbol \size{C^S_i} \le \denssymbol \alpha_i n/b$. Thus, we need to assign a total of
\[ \sum_{v \in B_{ij}} \lfloor \frac{w_v}{\denssymbol \alpha_i c} \rfloor \le \sum_{v \in B_{ij}} \frac{w_v}{\denssymbol \alpha_i c} =  \frac{1}{\denssymbol \alpha_i c} \sum_{v \in B_{ij}} w_v \le \frac{\denssymbol \alpha_i n/b }{\denssymbol \alpha_i c} = n/bc = a \]
nodes with $\sigma(v) \in B_{ik}$. Since all nodes $v \in V$ can broadcast $w_v$ and the assignments to node sets $B_{ik}$ are known globally, all nodes can construct the function $\sigma$ locally.

The rest of the algorithm now proceeds as follows:
\begin{enumerate}
    \item Apply Lemma~\ref{lemma:product-algorithm} with $\sigma$. This takes $O(\dens{S} a/n + \dens{T} b/n + 1)$ rounds.
    \item For each $(i,j,k) \in [b] \times [a] \times [c]$, each node that either originally computed the product $S[C_i^S, C^{ij}_k]T[C^{ij}_k, C^{T}_j]$ or was assigned it by $\sigma$ computes which entries in the product are non-zero entries of $\s{P_k}$; this is possible since all nodes that were assigned the product $S[C_i^S, C^{ij}_k]T[C^{ij}_k, C^{T}_j]$ are in $B_{ik}$, and thus by Lemma~\ref{lemma:filtering-bs} know the cutoff values for rows $r \in C^S_i$ of matrix $\s{P_k}$. Each of these nodes then assumes responsibility for $O(\denssymbol c)$ entries of $\s{P_k}$ from that product and discards the rest. This step can be done locally based on information obtained earlier.
\end{enumerate}
\end{proof}

%
%


\section{Distance tools}\label{section:dist-tools}

In this section, we use our matrix multiplication algorithms to construct basic distance computation tools that will be used for our final distance computation algorithms. Though we only use the distance tools for undirected graphs, we note that they work also for directed graphs. As noted before, we assume the edge weights are non-negative integers at most $O(n^c)$ for a constant $c$.

\subsection{Distance products}

\paragraph{Augmented min-plus semiring.}

The general algorithmic idea for our distance tools is to apply the matrix multiplication algorithms over the min-plus semiring. However, to ensure that we get consistent results in terms of hop distances, a property we require for our distance \knearest and \sdk distance tools,
we augment the basic min-plus semiring to keep track of the number of hops.

 We define the \emph{augmented min-plus semiring} $R$ to encode paths in distance computations as follows. The elements of $R$ are tuples $(w, t)$, where
\begin{enumerate}
    \item $w$ is either the weight of an edge or a path, or $\infty$, and
    \item $t$ is a non-negative integer or $\infty$, representing the number of hops.
\end{enumerate}
Let $\prec$ be the lexicographical order on tuples $(w,t)$, and define the addition operator $\min$ as the minimum over the total order given by $\prec$.
The multiplication operation $+$ is defined as $(w_1,t_1) + (w_2,t_2) = (w_1 + w_2, t_1 + t_2)$.
It is easy to verify that $\bigl(R, \min, +, (\infty,\infty), (0,0)\bigl)$ is a semiring with idempotent addition, that is, $\min(r,r) = r$ for all $r$. Moreover, the structure satisfies the conditions of Theorem~\ref{theorem:mm-filtered}.

\paragraph{Distance products.} We call the product of $S$ and $T$ over the augmented min-plus semiring the \emph{augmented distance product} of $S$ and $T$, and denote it by $S \star T$. In particular, for a graph $G = (V,E)$, we define the \emph{augmented weight matrix} $W$ by setting
\[ W[u,v] =
\begin{cases}
(0,0) & \text{if $u = v$,}\\
(w(u,v),1) & \text{if there is an edge from $u$ to $v$, and}\\
(\infty, \infty) & \text{otherwise.}
\end{cases}
\]
As with the regular distance product, the $d$th augmented distance product power $W^d$ gives the distances for all pairs of nodes $u,v \in V$ using paths of at most $d$ hops, as well as the associated number of hops.

Finally, we observe that the augmented distance product gives a consistent ordering in terms of distance from $v$, in the following sense:

\begin{lemma}
\label{claim:subShortestPath}
Let $v, u \in V$, and let $P$ be the shortest path of at most $d$ hops from $u$ to $v$. Then for every node $w$ on the path $P$, we have $W^d[v,w] \prec W^d[u,v]$.
\end{lemma}

\begin{proof}
It is sufficient to prove the claim for $w$ that is the last node on $P$ before $u$, and the rest follows by a straightforward induction. Assume towards a contradiction that $w \not\in N_k(v)$. If $d(w,u)>0$, then clearly $d(v,w)<d(v,u)$ because $d(v,u)=d(v,w)+d(w,u)$ since $P$ is a shortest path. If $d(w,u)=0$, then the hop distance from $v$ to $u$ is larger by $1$ compared to the hop distance between $v$ and $w$. In both cases, we have $W^n[v,w] \prec W^n[v,u]$.
\end{proof}

\paragraph{Recovering paths.} As noted in \cite{censor2015algebraic}, it is possible to recover a \emph{routing table} from the distance product algorithms in addition to the distances. Specifically, it is easy to see that since the matrix multiplication algorithms explicitly compute the non-zero products, they can be modified to provide a \emph{witness} for each non-zero entry in a product $P = ST$. That is, for any non-zero $P[u,v]$, we also get a witness $w_{uv} \in V$ such that $P[u,v] = S[u,w_{uv}] + T[w_{uv},v]$. This allows us to obtain, for any distance estimate $d(u,v)$, a node $w \in N(v)$ such that the shortest path from $v$ to $u$ uses the edge $(v,w)$, as discussed in \cite{censor2015algebraic}.

\subsection{$k$-nearest neighbors}

In the \knearest problem, we are given an integer $k$, and the task is to compute for each node $v$ the set of $k$ nearest nodes and distances to those nodes, breaking ties by first by hop distance and then arbitrarily. More formally, we want each node to compute a set $N_{k}(v)$ of $k$ nodes and distances $d(v,u)$ for all $u \in N_k(v)$, such that the values $W^n[v,u]$ for $u \in N_k(v)$ are the $k$ smallest on row $v$ of $W^n$ in terms of the order $\prec$ on augmented min-sum semiring $R$.

Note that it follows immediately from Lemma~\ref{claim:subShortestPath} that all nodes $u \in N_k(v)$ are at most $k$ hops away from $v$, and all nodes on the shortest path from $v$ to $u$ are also in $N_k(v)$.

\begin{theorem}\label{thrm:knearest}
The \knearest problem can be solved in 
\[O \left(\left( \frac{k}{n^{2/3} } + \log n \right) \log  k \right)\]
rounds in \clique.
\end{theorem}

\begin{proof}
For matrix $M$, let $\s{M}$ denote the matrix obtained by discarding all but the $k$ smallest values on each row of $M$. To solve the \knearest problem, we compute the filtered version $\s{W^k}$ of the $k$th power of the augmented weight matrix $W$ as follows:
\begin{enumerate}
    \item All nodes $v$ discard all but the $k$ smallest values on row $v$ to obtain $\s{W}$.
    \item We now observe that
    \[ \s{W^2} = \s{\s{W} \star \s{W}}, \qquad \s{W^4} = \s{\s{W^2} \star \s{W^2}}, \qquad \dots, \qquad \s{W^k} = \s{\s{W^{\frac{k}{2}}} \star \s{W^{\frac{k}{2}}}}\,.\]
    Thus, by applying Theorem~\ref{theorem:mm-filtered} with $\denssymbol = k$ iteratively $O(\log k)$ times, we can compute the matrix $\s{W^k}$ in $O \left(\left( \frac{k}{n^{2/3} } + \log n \right) \log  k \right)$ rounds.
\end{enumerate}
To see that this allows us to recover $N_k(v)$, we first observe that by Lemma~\ref{claim:subShortestPath}, the hop distance from $v$ to any node in $N_k(v)$ is at most $k$. Moreover, by a simple induction argument using Lemma~\ref{claim:subShortestPath}, we have that for all non-zero entries of $\s{W^k}$, we have $\s{W^k}[u,v] = W^k[u,v]$. Thus, the non-zero entries on row $v$ give us the set $N_k(v)$ and the distances to those nodes.
\end{proof}

\subsection{Source detection}

In the \sdk problem, we are given a set of \emph{sources} $S \subseteq V$ and integers $k$ and $d$, and the task is to compute for each node $v$ the set of $k$ nearest sources within $d$ hops, as well as the distances to those sources using paths of at most $d$ hops.

\begin{theorem}\label{thrm:source-detection}
The \sdk problem can be solved in
\[O \biggl(\biggl( \frac{m^{1/3}k^{2/3}}{n} + \log n \biggr) d \biggl) \qquad {\text{or}} \qquad O \biggl(\biggl( \frac{ m^{1/3} \size{S}^{2/3}}{n} +1 \biggr) d \biggl) \]
rounds in \clique, where $m$ is the number of edges in the input graph.
\end{theorem}

\begin{proof}
To obtain the first running time, we solve the \sdk problem as follows:
\begin{enumerate}
    \item All nodes, including nodes in $S$, select the $k$ lightest edges to nodes in $S$, breaking ties arbitrarily. Let $W_1$ be the augmented weight matrix restricted to these edges; since a total of at most $nk$ edges was selected, we have $\nz(W_1) \le nk$.
    \item We now apply Theorem~\ref{theorem:mm-filtered} iteratively $d$ times to compute the products
    \[ W_2 = \s{W \star W_1}, \qquad W_3 = \s{W \star W_2}, \qquad \dotsc, \qquad  W_d = \s{W \star W_{d-1}}\,. \]
    We have $\dens{W} = m/n$ and $\dens{W_i} = k$, and use $\denssymbol = k$ as the output density, so computing all the products takes $O \left(\left( \frac{m^{1/3}k^{2/3}}{n} + \log n \right) d \right)$ rounds.
\end{enumerate}
By a simple induction using Lemma~\ref{claim:subShortestPath}, we see that the non-zero entries on row $v$ of $W_i$ correspond to the $k$ nearest sources within $i$ hops of $v$.

For the second running time, we instead compute the $d$-hop distances from set $S$ to all other nodes:
\begin{enumerate}
    \item Let $U_1$ be the $n \times \size{S}$ matrix obtained by restricting the augmented weight matrix $W$ to edges with at least one endpoint in $S$. By padding matrix $U_1$ with zero entries, we can view it as a square matrix with density $\size{S}$.
    \item We now apply Theorem~\ref{theorem:mm} iteratively $d$ times to compute the products
    \[ U_2 = W \star U_1, \qquad U_3 = W \star U_2, \qquad \dotsc, \qquad  U_d = W \star U_{d-1}\,,\]
    where the density of $W$ is $m/n$ and the density of all matrices $U_i$ is $\size{S}$, giving a total running time of $O \left(\left( \frac{ m^{1/3} \size{S}^{2/3}}{n} +1 \right) d \right)$ rounds.
\end{enumerate}
The matrix $U_d$ gives the $d$-hop distances between nodes in $S$ and all other nodes, so each node can select the $k$ closest sources.
\end{proof}

\subsection{Distance through node set}

In the \distthrough problem, we assume that each node $v$ has a set $W_v$ and distance estimates $\delta(v,w)$ and $\delta(w,v)$ for all $w \in W_v$. The task is for all nodes $v$ to compute distance estimates
$\min_{w \in W_v \cap W_u } \{ \delta(v,w) + \delta(w,u) \}$
for all other nodes $u \in V$.

\begin{theorem}\label{thrm:distance-through}
The \distthrough problem can be solved
\[O \biggl( \frac{\denssymbol^{2/3}}{n^{1/3}} + 1 \biggr)\]
rounds in \clique, where $\denssymbol = \sum_{v \in V}\size{W_v}/n$.
\end{theorem}

\begin{proof}
Define matrices $W_1$ and $W_2$ as
\[ W_1[v,w] =
\begin{cases}
\delta(v,w)  & \text{if $w \in W_v$, and}\\
\infty & \text{if $w \notin W_v$,}
\end{cases}
\qquad\qquad
 W_2[w,v] =
\begin{cases}
\delta(w,v)  & \text{if $w \in W_v$, and}\\
\infty & \text{if $w \notin W_v$.}
\end{cases}
 \]
The distance product $W_1 \star W_2$ over the standard min-sum semiring clearly gives the desired estimates, and since $\dens{W_1} = \dens{W_2} = \delta$, it can be computed in $O \left( \frac{\denssymbol^{2/3}}{n^{1/3}} + 1 \right)$ rounds by Theorem~\ref{theorem:mm} (or by \cite{DBLP:conf/opodis/Censor-HillelLT18}), using $n$ as the density estimate for the output matrix.
\end{proof}

\section{Hopsets}\label{section:hopsets}


In this section, we describe a construction of hopsets in polylogarithmic time in the \clique model.
Given a graph $G=(V,E)$, a $(\beta,\epsilon)$-hopset $H=(V,E')$ is a graph on the same set of nodes such that the $\beta$-hop distances in $G \cup H$ give $(1+\epsilon)$-approximations for the distances in $G$. Formally, for any pair of nodes $u,v \in V$ it holds that $d_G(u,v) \leq d_{G \cup H}(u,v)$ and $$d_G(u,v) \leq d_{G \cup H}^{\beta}(u,v) \leq (1+\epsilon)d_G(u,v).$$
Usually the goal is to find a sparse hopset with small $\beta$. In our case, we are interested in optimizing both $\beta$ and the running time for constructing the hopset, but not necessarily the hopset size.

Our construction is based on the recent construction of hopsets in the \clique of Elkin and Neiman \cite{DBLP:journals/corr/ElkinN17}. However, the time complexity in \cite{DBLP:journals/corr/ElkinN17} depends on $\beta$ and on the hopset size, and is super-polylogarithmic for any choice of parameters. We show that using our new distance tools we can implement the same construction in polylogarithmic time, regardless of the hopset size, as long as $\beta$ is polylogarithmic. In particular, we focus on a simple variant with $\widetilde{O}(n^{3/2})$ edges which is enough for all our applications.\footnote{Using similar ideas, it is possible to implement also the more general hopset from \cite{DBLP:journals/corr/ElkinN17}.}

\subsection{Construction overview}

We follow the hopset construction of \cite{DBLP:journals/corr/ElkinN17}, which is based on the emulators of Thorup and Zwick \cite{thorup2006spanners}. A similar construction appears also in \cite{huang2019thorup} without a distributed implementation. We focus only on a simple variant of \cite{DBLP:journals/corr/ElkinN17,thorup2006spanners, huang2019thorup}, with slightly different parameters.

Given a graph $G=(V,E)$, let $V=A_0 \supseteq A_1 \supseteq A_2 = \emptyset$, where $A_1$ is a hitting set of size $O(\sqrt{n})$ of all the sets $N_k(v)$ for $k = O(\sqrt{n} \log{n})$. I.e., for any node $v \in V$, there is a node from $A_1$ among the closest $k$ nodes to $v$. We can construct $A_1$ using Lemma \ref{det_hit}.

For a given subset $A \subseteq V$, we denote by $d_G(v,A)$ the distance from $v$ to the closest node in $A$.
For a node $v \in V$, let $p(v) \in A_1$ be a node of distance $d_G(v,A_1)$ from $v$. In general there may be many possible nodes of distance $d_G(v,A_1)$ from $v$. The node $p(v)$ would be determined by $v$, as follows. During our algorithm, each node computes the set $N_k(v)$. Since $A_1$ is a hitting set there is a node from $A_1$ in the set $N_k(v)$ computed, and $v$ defines $p(v)$ to be the closest such node (breaking ties arbitrarily).
For a node $v \in A_0 \setminus A_1$, we define the \emph{bunch}
$$B(v) = \{u \in A_0: d_G(v,u) < d_G(v,A_1) \} \cup p(v),$$
and for a node $v \in A_1$, we define the bunch $B(v) = A_1$.

The hopset is the set of edges $H=\{\{v,u\}: v \in V, u \in B(v)\}$. We would like to set the length of an edge $\{v,u\}$ to be $d_G(v,u)$. However, we cannot necessarily compute these values. During our algorithm we add exactly all the edges in $H$ to the hopset, but their weights are not necessarily $d_G(v,u)$, but rather an approximation for $d_G(v,u)$.

\begin{claim} \label{edges_claim}
The number of edges in $H$ is $O(n^{3/2} \log{n})$.
\end{claim}

\begin{proof}
For every node in $A_1$ we add only edges to nodes in $A_1$, and the size of $A_1$ is $O(\sqrt{n})$. In addition, for nodes $v \in A_0 \setminus A_1$ we add edges only to nodes closer than $p(v)$. Since $A_1$ is a hitting set, $p(v)$ is among the closest $k$ nodes to $v$, which means that we add at most $k = O(\sqrt{n} \log{n})$ edges for each node, and $O(n^{3/2} \log{n})$ edges in total.
\end{proof}

\subsection{Congested clique implementation}

Our goal is to build $H$ efficiently in the \clique model.
In \cite{DBLP:journals/corr/ElkinN17}, the authors suggest an iterative algorithm for computing $H$, where in iteration $1 \leq \ell \leq \log{n}$ they compute a $2^{\ell}$-bounded hopset, which approximates the distances between all pairs of nodes that have a shortest path between them with at most $2^{\ell}$ edges in $G$. Given a $2^{\ell}$-bounded hopset $H^{\ell}$, they show that it is enough to run Bellman-Ford explorations up to hop-distance $O(\beta)$ in the graph $G \cup H^{\ell}$ to compute a $2^{\ell+1}$-hopset $H^{\ell+1}$. For the stretch analysis to carry on, it is also important to add edges to $H^{\ell+1}$ in a certain order. First, the edges that correspond to bunches of nodes in $A_0 \setminus A_1$ are added, and only then the edges that correspond to bunches of nodes in $A_1$ are added. They show that the total time complexity for implementing all the Bellman-Ford explorations is proportional to the hopset size. In particular, for our hopset of size $\widetilde{O}(n^{3/2})$, the complexity required is at least $\widetilde{O}(\sqrt{n})$ rounds.

We follow the general approach from \cite{DBLP:journals/corr/ElkinN17}, with two changes. First, we show that using our distance tools we can compute directly all the bunches of nodes in $A_0 \setminus A_1$, which allows simplifying slightly the algorithm and analysis. The second and more significant difference is that we use our $(S,d,k)$-algorithm to replace the heavy Bellman-Ford explorations, which results in a polylogarithmic complexity instead of a polynomial one. We next describe the algorithm in detail.

\subsubsection{Algorithm description}

\paragraph{Computing the bunches.}
We start by computing the bunches for all the nodes $v \in A_0 \setminus A_1$. This can be done easily using our algorithm for finding the distances to the $k$-nearest nodes. This follows since $p(v)$ is among the $k$-nearest nodes to $v$, hence all the nodes in the bunch of $v$ are among the $k$-nearest nodes to $v$. Note that all the nodes $v \in A_0 \setminus A_1$ learn the \emph{exact} distances to all the nodes in their bunch.
However, it is not clear how to compute a bunch for a node $v \in A_1$, since $v$ should learn the distances to all the nodes from $A_1$. Such nodes may be at distance $d = \Omega(n)$ from $v$, and our $(S,d,k)$-algorithm depends linearly on the distance $d$.
To overcome this, we follow the approach in \cite{DBLP:journals/corr/ElkinN17}, and show how to implement it efficiently using our distance tools.

\paragraph{Bounded hopsets.}
We say that $H$ is a $(\beta,\epsilon,t)$-hopset if for all $x,y \in G$ it holds that $d_G(x,y) \leq d_{G \cup H}(x,y)$, and for all pairs of nodes $x,y \in G$ that have a shortest path between them with at most $t$ hops, i.e., $d_G(x,y)=d_G^t(x,y)$, it holds that $$d_G(x,y) \leq d_{G \cup H}^{\beta}(x,y) \leq (1+\epsilon)d_G(x,y).$$

Note that the empty set is a $(1,0,1)$-hopset, and thus also a $(\beta,\epsilon,1)$-hopset for any $\beta \geq 1$ and $\epsilon > 0$.
The algorithm builds hopsets iteratively, where in iteration $\ell$ it builds a $(\beta,\epsilon_{\ell},2^{\ell})$-hopset $H^{\ell}$, where $\epsilon_{\ell}=\epsilon \cdot \ell$ for some $0< \epsilon < 1/ \log{n}$. The final hopset $H$ is $H^{\log{n}}$.
Let $H_0 =\{\{u,v\}: u \in A_0 \setminus A_1, v \in B(u)\}$. We already computed the edges $H_0$ and include them as part of the hopset for all the hopsets $H^{\ell}$.

\paragraph{Building $H^{\ell}$.}
We define $H^0 = H_0$. To construct $H^{\ell}$ for $1 \leq \ell \leq \log{n}$ we work as follows. Let $G' = G \cup H^{\ell -1}$.
All the nodes in $A_1$ compute the distances to all the nodes from $A_1$ at hop-distance at most $4\beta$ in the graph $G'$. We implement it using our $(S,d,k')$-algorithm with $S=A_1,d=4\beta,k'=|A_1| = O(\sqrt{n})$. Each node $v \in A_1$, adds to the hopset $H^{\ell}$ all the edges $\{v,u\}$ where $u \in A_1$ at hop-distance at most $4\beta$ from $v$ in $G'$. The weight of the edge is the weight that $v$ learned during the $(S,d,k)$-algorithm. The hopset $H^{\ell}$ includes all these edges, as well as all the edges of $H_0$. Note that if $v$ adds an edge $\{v,u\}$ to the hopset, in the next round it can let $u$ learn about it, so we can assume that both the endpoints know about the edge.

This completes the description of the algorithm.
We next analyze the time complexity and prove the correctness of the algorithm.

\subsubsection{Complexity}

\begin{claim} \label{hopsets_time_claim}
The complexity of the algorithm is $O(\beta \log{n} + \log^2{n})$ rounds.
\end{claim}

\begin{proof}
From Lemma \ref{det_hit}, we can construct deterministically in $O((\log{\log{n}})^3)$ rounds a hitting set $A_1$ of size $O(\sqrt{n})$ that hits all the sets $N_k(v)$. For the algorithm, each node $v$ should learn the set $N_k(v)$ which is computed using our algorithm for finding the distances to the $k$-nearest nodes. From Theorem~\ref{thrm:knearest}, this takes $O \Big( \Big( \frac{k}{n^{2/3} } + \log n \Big) \log k  \Big) = O \Big( \Big( \frac{\sqrt{n} \cdot \log{n}}{n^{2/3} } + \log{n} \Big) \log{n} \Big)=O(\log^2{n})$ rounds.

The algorithm for constructing $H$ starts by computing the bunches of nodes $v \in A_0 \setminus A_1$, this only requires each node $v$ to learn the distances to the nodes $N_k(v)$, which was already computed.
Next, for $\log{n}$ iterations, we compute for each node the $O(\sqrt{n})$-nearest nodes from $A_1$ at hop-distance at most $4\beta$. Since $A_1$ is of size $O(\sqrt{n})$, by Theorem~\ref{thrm:source-detection}, each iteration takes $O\Big( \Big( \frac{n^{2/3}{(\sqrt{n})}^{2/3}}{n} + 1 \Big) \cdot \beta \Big) = O(\beta)$ rounds. Overall, all the iterations take $O(\beta \log{n})$ rounds.
\end{proof}

\subsubsection{Correctness}

We next prove that $H^{\ell}$ is indeed a $(\beta,\epsilon_{\ell},2^{\ell})$-hopset, the proof follows the proof in \cite{DBLP:journals/corr/ElkinN17}, and is included here for completeness. There are slight changes to adapt it to the specific variant we consider. We start with a simple claim regarding the edges $H_0$.

\begin{claim} \label{basis}
For any $x,y \in V$, either $d_{G \cup H_0}^1(x,y) = d_G(x,y)$ or there exists $z \in A_1$ such that $d_{G \cup H_0}^1(x,z) \leq d_G(x,y).$
\end{claim}

\begin{proof}
If $x \in A_1$, the claim clearly holds by setting $z=x$. We next focus on $x \in A_0 \setminus A_1$. For each $x \in A_0 \setminus A_1$, $H_0$ includes the edge $\{x,p(x)\}$, as well as all the edges of the form $\{x,y\}$ where $d_G(x,y) < d_G(x,p(x))$.
Hence, if $d_G(x,y) < d_G(x,p(x))$, there is an edge $\{x,y\}$ in $H_0$ of weight $d_G(x,y)$ and we are done. Otherwise, $d_G(x,p(x)) \leq d_G(x,y)$ and the edge $\{x,p(x)\}$ of weight $d_G(x,p(x))$ is in $H_0$ where $p(x) \in A_1$ which completes the proof.
\end{proof}

The next lemma shows that $H^{\ell}$ is a $(\beta,\epsilon_{\ell},2^{\ell})$-hopset.
\begin{lemma} \label{hopset_lemma}
Let $0< \epsilon < 1/\log{n}$. Set $\delta = \epsilon / 4$, and $\beta=3/\delta$ and let $x,y \in V$ be such that $d_G(x,y)=d_G^{2^{\ell}}(x,y)$. Then, $$d_{G \cup H^{\ell}}^{\beta}(x,y) \leq (1+\epsilon_{\ell})d_G(x,y),$$ where $\epsilon_{\ell}=\epsilon \cdot \ell$.
\end{lemma}

\begin{proof}
The proof is by induction on $\ell$. For $\ell = 0$, the claim holds since the empty set is a $(1,0,1)$-hopset. Assume it holds for $\ell-1$, and we prove it for $\ell$. Let $x,y \in V$ be such that $d_G(x,y)=d_G^{2^{\ell}}(x,y)$, and let $\pi(x,y)$ be a shortest path between $x$ and $y$ in $G$ with at most $2^{\ell}$ edges. We partition $\pi(x,y)$ into $J \leq 1/\delta$ segments $\{L_j = [u_j,v_j]\}_{1 \leq j \leq J}$, each of length at most $\delta \cdot d_G(x,y)$, and additionally at most $1/\delta$ edges $\{v_j,u_{j+1}\}_{1 \leq j \leq J}$ between these segments. This is done as follows.

Define $u_1 = x$, and walk from $u_1$ on $\pi(x,y)$ towards $y$ until the first node $u_2$ such that $d_G(u_1,u_2) > \delta \cdot d_G(x,y)$, if such exists or until $y$. The last node in $\pi(x,y)$ before $u_2$ is $v_1$, and from the definition of $u_2$ it holds that $d_G(u_1,v_1) \leq \delta \cdot d_G(x,y)$. We continue in the same manner (walk from $u_2$ on $\pi(x,y)$ towards $y$ to define $u_3$, and so on) and define the rest of the nodes $u_j,v_j$. If we walk from $u_j$ towards $y$ and there is no node $u_{j+1}$ where $d_G(u_j,u_{j+1}) > \delta \cdot d_G(x,y)$, we define $v_j=y$, and $J=j$. Since $d_G(u_i,u_{i+1})> \delta \cdot d_G(x,y)$, the number of segments is at most $1/\delta$.

Now, we use Claim \ref{basis} on the pairs $\{u_j,v_j\}_{j=1}^{J}$. One case is that for all $1 \leq j \leq J$, it holds that $d_{G \cup H_0}^1(u_j,v_j) = d_G(u_j,v_j)$, which means that the edges $\{u_j,v_j\}$ exist in $G \cup H_0$ with weights $d_G(u_j,v_j)$. This means that in the graph $G \cup H_0$ there is a path of hop-distance at most $2/\delta$ between $x$ and $y$ that includes all the edges of the form $\{u_j,v_j\}_{j=1}^{J}$ and the edges $\{v_j,u_{j+1}\}$ between segments. The total weight of the path is exactly $d_G(x,y)$. Since $\beta = 3/\delta > 2/\delta$, and $H_0 \subseteq H^{\ell}$, we get that $d_{G \cup H^{\ell}}^{\beta}(x,y) = d_G(x,y) \leq (1+\epsilon_{\ell})d_G(x,y)$.

We now analyse the case where there is at least one pair $\{u_j,v_j\}$ that does not satisfy $d_{G \cup H_0}^1(u_j,v_j) = d_G(u_j,v_j)$. Let $i,j$ be the indexes of the first and last pairs $\{u_i,v_i\},\{u_j,v_j\}$ that do not satisfy the above. From Claim \ref{basis}, there are nodes $z_i,z_j \in A_1$ such that $$d_{G \cup H_0}^1(u_i,z_i) \leq d_G(u_i,v_i) \leq \delta \cdot d_G(x,y),$$ $$d_{G \cup H_0}^1(v_j,z_j) \leq d_G(u_j,v_j) \leq \delta \cdot d_G(x,y).$$ In our algorithm, when constructing $H^{\ell}$, the nodes $z_i,z_j$ add edges to all the nodes of $A_1$ at hop-distance at most $4\beta$ in the graph $G' = G \cup H^{\ell-1}$. We next show that in $G'$ there is a path of hop-distance at most $4\beta$ between $z_i,z_j$ and bound its weight. From the induction hypothesis, $H^{\ell-1}$ is a $(\beta, \epsilon_{\ell-1},2^{\ell-1})$-hopset. Since $u_i,v_j$ are in the shortest path between $x$ and $y$, it holds that $d_G(u_i,v_j)=d_G^{2^{\ell}}(u_i,v_j)$. Since any path of hop-distance at most $2^{\ell}$ can be partitioned to two paths of hop-distance at most $2^{\ell-1}$, it holds that a $(\beta, \epsilon_{\ell-1},2^{\ell-1})$-hopset is also a $(2\beta, \epsilon_{\ell-1},2^{\ell})$ hopset. Hence, there is a path of hop-distance at most $2\beta$ between $u_i$ and $v_j$ in $G \cup H^{\ell-1}$ of weight at most $(1+\epsilon_{\ell-1})d_G(u_i,v_j)$. Since the edges $\{u_i,z_i\},\{v_j,z_j\}$ of weight at most $\delta \cdot d_G(x,y)$ are in $H_0$, we get that there is a path of hop-distance at most $2\beta +2$ between $z_i$ and $z_j$ in $G'$ of total weight at most $(1+\epsilon_{\ell-1})d_G(u_i,v_j) + 2\delta \cdot d_G(x,y)$. Since $2\beta+2 \leq 4\beta$, $z_i$ and $z_j$ add an edge between them to $H^{\ell}$, which gives $$d^1_{G \cup H^{\ell}}(z_i,z_j) \leq (1+\epsilon_{\ell-1})d_G(u_i,v_j) + 2\delta \cdot d_G(x,y).$$

We next bound $d_{G \cup H^{\ell}}^{\beta}(x,y)$. From the definition of $i,j$, for any $i' < i$ or $j < i'$, there is an edge $\{u_{i'},v_{i'}\}$ in $H_0$ (and hence in $H^{\ell}$) of weight $d_G(u_{i'},v_{i'})$. These edges together with the edges in $\pi(x,y)$ before or after these segments sum up to at most $2/\delta$ edges. Between $u_i$ and $v_j$ we have a 3-hop path $\{u_i,z_i,z_j,v_j\}$ of total weight at most $\delta \cdot d_G(x,y) + d^1_{G \cup H^{\ell}}(z_i,z_j) + \delta \cdot d_G(x,y) \leq (1+\epsilon_{\ell-1})d_G(u_i,v_j) + 4\delta \cdot d_G(x,y)$. To conclude, $G \cup H^{\ell}$ has a path of hop-distance at most $2/\delta +3 \leq \beta$ of total weight $$d_G(x,u_i)+(1+\epsilon_{\ell-1})d_G(u_i,v_j) + d_G(v_j,y) + 4\delta \cdot d_G(x,y) \leq (1+\epsilon_{\ell-1} +4\delta)d_G(x,y).$$ By the choice of parameters $4\delta=\epsilon$, which gives $d_{G \cup H^{\ell}}^{\beta}(x,y) \leq (1+\epsilon_{\ell})d_G(x,y)$. This completes the proof.
\end{proof}

\subsubsection{Conclusion}

Using Claims \ref{edges_claim} and \ref{hopsets_time_claim} and Lemma \ref{hopset_lemma}, we obtain the following:

\begin{theorem}
\label{hopset_thm}
Let $0< \epsilon < 1$.
There is a deterministic construction of a $(\beta,\epsilon)$-hopset with $O(n^{3/2}\log{n})$ edges and $\beta=O\bigl((\log{n})/\epsilon\bigr)$ that takes $O\bigl((\log^2{n})/\epsilon \bigr)$ rounds in the \clique model.
\end{theorem}

\begin{proof}
By Lemma \ref{hopset_lemma}, $H^{\log{n}}$ is a $(\beta,\epsilon_{\log{n}},n)$-hopset, and hence also a $(\beta, \epsilon_{\log{n}})$-hopset. From the choice of parameters $\epsilon_{\log{n}} = \epsilon \log{n}$, where $0<\epsilon< 1/\log{n}$ and $\beta=O(1/\epsilon)$. Let $\epsilon' = \epsilon_{\log{n}} = \epsilon \log{n}$. Then, $\epsilon = \frac{\epsilon'}{\log{n}}$. Hence, $H^{\log{n}}$ is a $(\beta,\epsilon')$-hopset where $0< \epsilon' < 1$ and $\beta=O(\frac{\log{n}}{\epsilon'})$. 
By Claim \ref{hopsets_time_claim}, the time complexity for constructing the hopset is $O(\beta \log{n} + \log^2{n})$. Since $\beta=O(\frac{\log{n}}{\epsilon'})$, the complexity is $O(\frac{\log^2{n}}{\epsilon'})$ rounds. The number of edges is $O(n^{3/2}\log{n})$ by Claim \ref{edges_claim}.
\end{proof}

\section{Multi-source shortest paths}\label{section:MSSP}

As a direct application of our hopsets and source detection algorithms, we get an efficient $(1+\epsilon)$-approximation for the multi-source shortest paths problem (MSSP) with polylogarithmic complexity as long as the number of sources is $\widetilde{O}(\sqrt{n})$.

\mssp*

\begin{proof}
By Theorem \ref{hopset_thm}, we can construct a $(\beta,\epsilon)$-hopset $H$ with $\beta=O(\frac{\log{n}}{\epsilon})$ in $O(\frac{\log^2{n}}{\epsilon})$ rounds. By the definition of hopsets, the $\beta$-hop distances in $G \cup H$ give a $(1+\epsilon)$-approximation to all the distances in $G$. Hence, to approximate the distances of all the nodes to the set of sources $S$, we run our $(S,d,k)$-algorithm in the graph $G \cup H$ with $d = \beta = O(\frac{\log{n}}{\epsilon})$. By Theorem~\ref{thrm:source-detection}, the complexity is $O \Big( \Big ( \frac{n^{2/3}{|S|}^{2/3}}{n} +1 \Big ) \cdot \beta \Big) =  O \Big( \Big ( \frac{{|S|}^{2/3}}{n^{1/3}} +1 \Big ) \cdot \frac{\log{n}}{\epsilon} \Big)$ rounds. The overall complexity is $$O \bigg( \bigg ( \frac{{|S|}^{2/3}}{n^{1/3}} +1 \bigg ) \cdot \frac{\log{n}}{\epsilon} + \frac{\log^2{n}}{\epsilon} \bigg) = O \bigg( \bigg( \frac{{|S|}^{2/3}}{n^{1/3}} +\log{n} \bigg) \cdot \frac{\log{n}}{\epsilon} \bigg)$$ rounds.
\end{proof}

\section{All-pairs shortest paths} \label{sec:APSP}

In this section, we use our multi-source shortest paths algorithm from $\widetilde{O}(\sqrt{n})$ sources to construct algorithms that provide constant approximations for \emph{all} the distances in the graph in \emph{polylogarithmic} time. However, the approximation guarantee increases to $(2+\epsilon)$ for the unweighted case, and at most $(3+\epsilon)$ for the weighted case.

Next, we present our approximation algorithms for the APSP problem.
For simplicity of presentation, we start by sketching a simple $(3+\epsilon)$-approximation for weighted APSP.
Using a more careful analysis we show that a variant of this algorithm actually obtains a $(2+\epsilon,(1+\epsilon)W)$-approximation. Here, we use the notation $(2+\epsilon,(1+\epsilon)W)$-approximation for an algorithm that for any pair of nodes $u,v$ provides an estimate of the distance between $u$ and $v$ which is at most $(2+\epsilon)d(u,v) + (1+\epsilon)W$ where $W$ is the weight of the heaviest edge in a shortest path between $u$ and $v$. Note that $W$ is always at most $d(u,v)$, and hence a $(2+\epsilon,(1+\epsilon)W)$-approximation is always better than a $(3+\epsilon')$-approximation (where $\epsilon' = 2\epsilon$).
Finally, we extend the algorithm to obtain a $(2+\epsilon)$-approximation for APSP in unweighted graphs.
All our algorithms are deterministic.

\subsection{$(3+\epsilon)$-approximation for weighted APSP}

Our algorithm starts by computing a hitting set $A$ of size $\widetilde{O}(\sqrt{n})$ that hits all the sets $N_k(v)$ for $k=\widetilde{O}(\sqrt{n})$. We compute for each node $v \in V$, the distances to all the nodes in $N_k(v)$, and we use our MSSP algorithm to compute $(1+\epsilon)$-approximate distances from $A$ to all the nodes. Since $k = \widetilde{O}(\sqrt{n})$ and $|A|=\widetilde{O}(\sqrt{n})$ the complexity is polylogarithmic. To compute an estimate $\delta(u,v)$ of the distance $d(u,v)$, we work as follows. If $v \in N_k(u)$, then $u$ already computed the value $d(u,v)$ and sends it to $v$. Otherwise, let $p(u) \in N_k(u) \cap A$ be a closest node to $u$. We estimate $d(u,v)$ with $d(u,p(u))+d(p(u),v).$
Since $v \not \in N_k(u)$, it follows that $d(u,p(u)) \leq d(u,v).$ Hence, using the triangle inequality, we have $$d(u,p(u)) + d(p(u),v) \leq d(u,p(u)) + d(u,p(u)) + d(u,v) \leq 3d(u,v).$$ Since our algorithm computes an $(1+\epsilon)$-approximation for the value $d(v,p(u))$, we get a $(3+\epsilon)$-approximation.

\subsection{$(2+\epsilon,(1+\epsilon)W)$-approximation for weighted APSP} \label{section:wAPSP}

We next describe our $(2+\epsilon,(1+\epsilon)W)$-approximation. In our $(3+\epsilon)$-approximation, we used the fact that if $v \not \in N_k(u)$, then $d(u,p(u)) \leq d(u,v)$. To obtain a better approximation, we would like to get a better bound on $d(u,p(u))$. If, for example, $d(u,p(u)) \leq \frac{1}{2} d(u,v)$, the same analysis shows a $(2+\epsilon)$-approximation. We show that either $\min\{ d(u,p(u)),d(v,p(v)) \} \leq \frac{1}{2}(d(u,v)+W)$ which proves a $(2+\epsilon,(1+\epsilon)W)$-approximation, or the shortest path between $u$ and $v$ has a node $w \in N_k(u) \cap N_k(v)$. In the latter case, $d(u,w)+d(w,v)=d(u,v)$, and since $|N_k(u)|=|N_k(v)|=\widetilde{O}(\sqrt{n})$ we can use matrix multiplication to compute the distances through all nodes $w \in N_k(u) \cap N_k(v),$ and find the actual distance $d(u,v)$.

\subsubsection{Algorithm description}

We next describe the algorithm in detail. We start by describing the general structure of the algorithm, and then explain how we implement it efficiently. In the algorithm, each node $v \in V$, maintains an estimate $\delta(u,v)$ of the distance $d(u,v)$ for all $u \in V$. Each time a node updates the estimate $\delta(u,v)$ it takes the minimum value between the previous value of $\delta(u,v)$ and the current estimate computed. Also, if $v$ updates $\delta(u,v)$ it sends the new estimate also to $u$ to update its estimate accordingly.
We use the notation $w(u,v)$ to denote the weight of the edge $\{u,v\}$ if exists. The distances $d(u,v)$ and weights $w(u,v)$ are with respect to the input graph $G$.
The algorithm works as follows.

\begin{oframed}
\begin{enumerate}
\item Set $\delta(u,v)=w(u,v)$ if $\{u,v\} \in E$, and $\delta(u,v) = \infty$ otherwise.\label{alg_init}
\item For each $v$, compute the distances to the set $N_k(v)$ for $k = \sqrt{n},$ and update the estimates $\{\delta(v,u)\}_{u \in N_k(v)}$ accordingly.\label{alg_nearest}
\item Set $\delta(u,v) = \min\{\delta(u,v), \min_{w \in N_k(u) \cap N_k(v)} \{\delta(u,w) + \delta(w,v)\}\}.$\label{alg_w}
\item Compute a hitting set $A$ of size $O(\sqrt{n} \log{n})$ that hits all the sets $N_k(v).$\label{alg_hit}
\item Use the MSSP algorithm to compute $(1+\epsilon)$-approximate distances from $V$ to $A$. For each $v \in V$ update the estimates $\{\delta(v,u)\}_{u \in A}$ accordingly.\label{alg_mssp}
\item Let $p(v) \in N_k(v) \cap A$ be a node from $A$ closest to $v$ that is chosen by $v$. The node $v$ sends $p(v)$ to all the nodes.\label{alg_p}
\item Set $\delta(u,v) = \min\{\delta(u,v), \delta(u,p(u))+\delta(p(u),v), \delta(v,p(v))+\delta(p(v),u)\}.$\label{alg_dis}
\end{enumerate}
\end{oframed}

\subsubsection{Complexity}

\begin{lemma} \label{wAPSP_time}
The complexity of the algorithm is  $O(\frac{\log^2{n}}{\epsilon})$ rounds.
\end{lemma}

\begin{proof}
Implementing Line \ref{alg_init} is immediate.
We implement Line \ref{alg_nearest} using our $k$-nearest algorithm, which takes $O \Big( \Big( \frac{k}{n^{2/3}} + \log n \Big) \log k \Big)=O(\log^2{n})$ rounds by Theorem~\ref{thrm:knearest}. We implement Line \ref{alg_w} with our distance through sets algorithm, with respect to the sets $N_k(v)$. Since $k=\sqrt{n}$, it takes $O\Big(\frac{k^{2/3}}{n^{1/3}}+1 \Big)=O(1)$ rounds by Theorem \ref{thrm:distance-through}. We use Lemma \ref{det_hit} to compute the hitting set $A$ in $O((\log{\log{n}})^3)$ rounds (note that all the nodes already computed the sets $N_k(v)$).  Each node updates all the nodes if it is in $A$ which takes one round.
We implement Line \ref{alg_mssp} using our MSSP algorithm from Theorem \ref{multi-thm}, which takes $O(\frac{\log^2{n}}{\epsilon})$ rounds since $|A| = O(\sqrt{n} \log{n})$. Implementing Line \ref{alg_p} is immediate since each node $v$ knows $N_k(v)$ and $A$. Then, all the nodes know all the values $\{p(v)\}_{v \in V}$. Now, each node $u$ sends to each node $v$ the values $\{d(u,p(u)),d(u,p(v))\}$. This allows $u$ and $v$ to compute the values $\delta(u,p(u))+\delta(p(u),v)$ and $\delta(v,p(v))+\delta(p(v),u)$, which allows them to compute the final estimate $\delta(u,v).$
The overall complexity is $O(\frac{\log^2{n}}{\epsilon})$ rounds.
\end{proof}

\subsubsection{Correctness}

\begin{lemma} \label{wAPSP_approx}
The algorithm gives a $(2+\epsilon,(1+\epsilon)W)$-approximation for weighted APSP.
\end{lemma}

\begin{proof}
Let $u,v \in V$. If $v \in N_k(u)$ or $u \in N_k(v)$, then after Line \ref{alg_nearest}, $\delta(u,v) = d(u,v)$ and we are done.
Otherwise, let $P$ be a shortest path between $u$ and $v$. Let $u'$ be the furthest node from $u$ in $P \cap N_k(u)$, and let $v'$ be the furthest node from $v$ in $P \cap N_k(v)$. We divide the analysis to 3 cases. See Figure \ref{APSP_pic} for an illustration.

\begin{figure}[h]
\centering
\includegraphics[width=\linewidth]{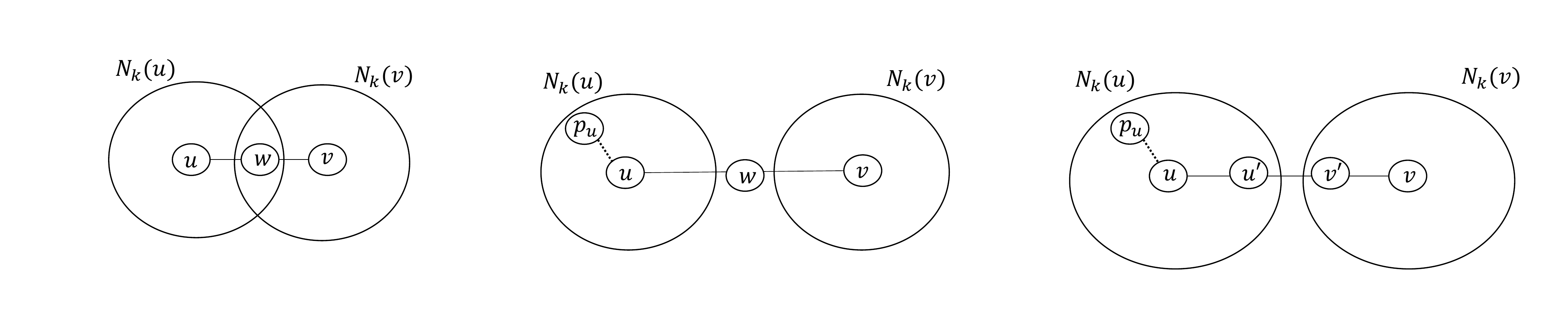}
\caption{There are 3 cases, the first is that $u'$ or $v'$ are in $N_k(u) \cap N_k(v)$. The second is that there is a node $w \in P \setminus (N_k(u) \cup N_k(v))$ and the last where there is an edge of $P$ between $u'$ and $v'$.}
\label{APSP_pic}
\vspace{0.2em}
\end{figure}

\textbf{Case 1:} Either $u' \in N_k(v)$ or $v' \in N_k(u)$.
In this case, $\min_{w \in  N_k(u) \cap N_k(v)}\{\delta(u,w)+\delta(w,v)\} = d(u,v)$, hence after Line \ref{alg_w}, $\delta(u,v) = d(u,v)$ and we are done.

\textbf{Case 2:} There is a node $w \in P$ where $w \not \in N_k(u)$ and $w \not \in N_k(v)$.
In this case, it follows that $d(u,p(u)) \leq d(u,w)$ and $d(v,p(v)) \leq d(v,w)$. Since $w \in P$, we also have $d(u,v) = d(u,w)+d(v,w)$. Hence $$\min\{d(u,p(u),d(v,p(v)\} \leq \min\{d(u,w),d(v,w)\} \leq d(u,v)/2.$$ Assume w.l.o.g that $d(u,p(u)) \leq d(v,p(v))$, then $d(u,p(u)) \leq d(u,v)/2$ and $d(v,p(u)) \leq d(v,u)+d(u,p(u)) \leq \frac{3}{2} d(u,v)$, which gives $d(u,p(u))+d(p(u),v) \leq 2d(u,v).$ The value $d(u,p(u))$ is computed exactly in Line \ref{alg_nearest}, and the value $d(v,p(u))$ is approximated within a $(1+\epsilon)$ factor in Line \ref{alg_mssp}. Hence, the final estimate is $\delta(u,v) \leq (2+\epsilon')d(u,v)$ for $\epsilon' = 2\epsilon$.

\textbf{Case 3:} All the nodes in $P$ are either in $N_k(u)$ or $N_k(v)$, but $N_k(u) \cap N_k(v) = \emptyset.$ In this case, the path $P$ is composed of a path between $u$ and $u'$, an edge between $u'$ and $v'$ and a path between $v'$ and $v$. Since $v' \not \in N_k(u)$ and $u' \not \in N_k(v)$, we have $d(u,p(u)) \leq d(u,v')$ and $d(v,p(v)) \leq d(v,u')$. In addition, since $d(u,v) = d(u,u') + w(u',v') + d(v',v)$, we have $$d(u,v')+d(v,u') \leq d(u,u') + w(u',v') + d(v,v') + w(u',v') \leq d(u,v) + w(u',v').$$ This implies that $$\min\{d(u,p(u)),d(v,p(v))\} \leq \min\{d(u,v'),d(v,u')\} \leq \frac{1}{2}(d(u,v) + w(u',v')).$$ Assume w.l.o.g that $d(u,p(u)) \leq d(v,p(v))$, then
$$d(u,p(u))+d(p(u),v) \leq d(u,p(u))+d(u,p(u))+d(u,v) \leq 2d(u,v) + w(u',v') \leq 2d(u,v)+W.$$
Since we compute the value $d(u,p(u))$ exactly, and approximate the value $d(p(u),v)$ within a $(1+\epsilon)$ factor, the final estimate is $\delta(u,v) \leq (2+2\epsilon)d(u,v)+(1+\epsilon)W.$ Choosing $\epsilon' = 2\epsilon$ completes the proof.
\end{proof}

\subsubsection{Conclusion}

From Lemmas \ref{wAPSP_time} and \ref{wAPSP_approx}, we get the following.

\begin{theorem}
There is a deterministic $(2+\epsilon,(1+\epsilon)W)$-approximation algorithm for weighted APSP in the \clique model that takes $O(\frac{\log^2{n}}{\epsilon})$ rounds.
\end{theorem}

\subsection{$(2+\epsilon)$-approximation for unweighted APSP} \label{section:uAPSP}

In this section, we present our $(2+\epsilon)$-approximation for unweighted APSP.

\subsubsection{High-level idea}

In our $(2+\epsilon,1+\epsilon)$-approximation, the only case where we do not obtain a $(2+\epsilon)$-approximation is when the shortest path between $u$ and $v$ is composed of a path between $u$ and $u' \in N_k(u) \setminus N_k(v)$, an edge between $u'$ and a node $v' \in N_k(v) \setminus N_k(u)$ and a path between $v'$ and $v$. Note that $u$ knows the distance to $u'$, $v$ knows the distance to $v'$ and both $u'$ and $v'$ know about the edge between them. If we add to the graph all the edges between $u$ and nodes in $N_k(u)$ and between $v$ and nodes in $N_k(v)$, we have a shortest path of length 3 between $u$ and $v$ in the new graph. Intuitively, we would like to use matrix multiplication to learn about this path, but we cannot do this directly since this graph may be dense. In particular, we have no bound on the degrees of $u',v'$.

To overcome this, we present a new algorithm that takes care separately of shortest paths that include a high-degree node, and shortest paths that have only low-degree nodes, where a node has a high-degree if its degree is $\widetilde{O}(\sqrt{n})$. This is done as follows. We compute a hitting set $A$ of size $\widetilde{O}(\sqrt{n})$, such that each high-degree node has a neighbour in $A$. Then, we use the fact that $|A|=\widetilde{O}(\sqrt{n})$, to show that we can compute efficiently $(1+\epsilon)$-approximations of all the shortest paths that go through a node in $A$. Then, if the shortest path between $u$ and $v$ has a high-degree node, it has a neighbour $w \in A$, and the shortest path from $u$ to $v$ that goes through $w$ already gives a good approximation for the shortest $u-v$ path.

Then, we are left only with paths that do not contain any high-degree node. The graph induced on low-degree nodes is sparse, since the degrees are at most $\widetilde{O}(\sqrt{n})$, and we would like to exploit it. However, although we can multiply two sparse matrices with $\widetilde{O}(n^{3/2})$ elements efficiently, to compute paths of length 3 we need to multiply 3 matrices of this size, which is no longer efficient. To overcome this, we would like to reduce the degrees further. In particular, we would like to follow the analysis in Lemma \ref{wAPSP_approx}, but replace $N_k(v)$ with $N_{k'}(v)$ for $k' < k$ which allows focusing on sparser matrices. The obstacle is that we want to compute a hitting set $A'$ of the sets $\{N_{k'}(v)\}_{v \in V}$, and then use our MSSP algorithm to compute distances to all the nodes in $A'$, but if $k' < k$, we need to compute a larger hitting set $A'$, and our MSSP algorithm is efficient only as long as $|A'| = \widetilde{O}(\sqrt{n})$.

Nevertheless, since our graph is sparse, we show that we can actually compute MSSP from a set $A'$ of size $\widetilde{O}(n^{3/4})$ which allows us to choose $k'=\widetilde{O}(n^{1/4})$. To deal with the problematic case in the proof of Lemma \ref{wAPSP_approx}, we should compute paths of length 3 of the form $\{u,u',v',v\}$ where $u' \in N_{k'}(u)$, $v' \in N_{k'}(v)$ and there is an edge between $u'$ and $v'$. Since $|N_{k'}(u)|= |N_{k'}(u)| = \widetilde{O}(n^{1/4})$ and $u',v'$ are in the sparse graph and have maximum degree $\widetilde{O}(\sqrt{n})$, we can detect paths of this form efficiently, which allows us to show a $(2+\epsilon)$-approximation.

\subsubsection{Algorithm description}

We next present a pseudo-code of the algorithm, then we show how to implement it efficiently and provide a correctness proof. Again, each node $v \in V$, maintains an estimate $\delta(u,v)$ of the distance $d(u,v)$ for all $u \in V$. Each time a node updates the estimate $\delta(u,v)$ it takes the minimum value between the previous value of $\delta(u,v)$ and the current estimate computed. Also, if $v$ updates $\delta(u,v)$ it sends the new estimate also to $u$ to update its estimate accordingly.
Let $N(v)=\{u \in V| \{v,u\} \in E \} \cup \{v\}.$
Our algorithm works as follows.

\begin{oframed}
\begin{enumerate}
\item Set $\delta(u,v) = w(u,v)$ if $\{u,v\} \in E$ and $\delta(u,v) = \infty$ otherwise.\label{ualg_init}\\

\textbf{First phase -} handling shortest paths with a high-degree node:

\item Let $A$ be a hitting set of size $\widetilde{O}(\sqrt{n})$ that hits all the neighbourhoods $N(v)$ of high-degree nodes with $|N(v)| \geq k$ for $k=\widetilde{O}(\sqrt{n})$.\label{ualg_hit}
\item Use the MSSP algorithm to compute $(1+\epsilon)$-approximate distances from $V$ to $A$. For each $v \in V$ update the estimates $\{\delta(v,u)\}_{u \in A}$ accordingly.\label{ualg_mssp}
\item Set $\delta(u,v) = \min\{\delta(u,v), \min_{w \in A} \{ \delta(u,w) + \delta(w,v) \}\}$.\label{ualg_A}\\

\textbf{Second phase -} handling shortest paths that contain only low-degree nodes:

Let $G'=(V',E')$ be the graph induced on low-degree nodes with degree at most $k$, from now on our algorithm works only in the graph $G'$. In particular, the set $N_{k'}(v)$ is the set of $k'$ closest nodes to $v$ in $G'$.

\item For each $v$, compute the distances to the set $N_{k'}(v)$ for $k'=\widetilde{O}(n^{1/4}),$ and update the estimates $\{\delta(v,u)\}_{u \in N_{k'}(v)}$ accordingly.\label{ualg_k'nearest}
\item Set $\delta(u,v) = \min\{\delta(u,v), \min_{w \in N_{k'}(u) \cap N_{k'}(v)} \{\delta(u,w) + \delta(w,v)\}\}.$\label{ualg_w}
\item Construct a hitting set $A'$ of of size $\widetilde{O}(n^{3/4})$ that hits all the sets $N_{k'}(v)$.\label{ualg_A'}
\item Compute $(1+\epsilon)$-approximate shortest paths from $A'$ to all the nodes in the graph $G'$. For each $v \in V$ update the estimates $\{\delta(v,u)\}_{u \in A'}$ accordingly.\label{ualg_mssp'}
\item Let $p'(v) \in N_{k'}(v) \cap A'$ be a node from $A'$ closest to $v$ that is chosen by $v$. The node $v$ sends $p'(v)$ to all the nodes.\label{ualg_p'}
\item Set $\delta(u,v) = \min\{\delta(u,v), \delta(u,p'(u))+\delta(p'(u),v), \delta(v,p'(v))+\delta(p'(v),u)\}$.\label{ualg_dist_p'}
\item Let $\delta'(u,v) = \min \{\delta(u,u')+\delta(u',v')+\delta(v',v)| {u' \in N_{k'}(u),v' \in N_{k'}(v), \{u',v'\} \in E'}\}$.\label{ualg_3paths}
\item Set $\delta(u,v) = \min \{\delta(u,v), \delta'(u,v)\}.$
\end{enumerate}
\end{oframed}

\subsubsection{Complexity}

\begin{lemma} \label{APSP_time}
The complexity of the algorithm is $O(\frac{\log^2{n}}{\epsilon})$ rounds.
\end{lemma}

\begin{proof}
Implementing Line \ref{ualg_init} is immediate.
We choose $k=\sqrt{n}$. To implement Line \ref{ualg_hit}, we use Lemma \ref{det_hit} to build $A$ of size $O(n\log{n}/k)=O(\sqrt{n} \log{n})$ that hits all the sets $N(v)$ for $v$ with $|N(v)|\geq k$. Note that each node $v$ knows the set $N(v)$. The complexity is $O((\log{\log{n}})^3)$ rounds.
Each node updates all the nodes if it is in $A$ which takes one round.
We implement Line \ref{ualg_mssp} using our MSSP algorithm from Theorem \ref{multi-thm}, which takes $O(\frac{\log^2{n}}{\epsilon})$ rounds since $|A| = O(\sqrt{n} \log{n})$. Line \ref{ualg_A} is equivalent to computing distances through the set $A$ with respect to the distances $\{\delta(u,v)\}_{u \in V, v \in A}$ computed in Line \ref{ualg_mssp}. This takes $O \Big( \frac{(\sqrt{n}\log{n})^{2/3}}{n^{1/3}} + 1) = O(\log{n})$ rounds by Theorem \ref{thrm:distance-through}.

To implement Line \ref{ualg_k'nearest}, we choose $k'=n^{1/4}$, and use Theorem~\ref{thrm:knearest} to find the $k'$-nearest nodes in $O \Big( \Big( \frac{k'}{n^{2/3}} + \log{n} \Big) \log k' \Big)=O(\log^2{n})$ rounds. To implement Line \ref{ualg_w}, we use our distance through sets algorithm, with respect to the sets $N_{k'}(v)$. Since $k'=n^{1/4}$, the complexity is $O(1)$ by Theorem \ref{thrm:distance-through}.
To implement Line \ref{ualg_A'}, we use Lemma \ref{det_hit} to build $A'$ of size $O(n\log{n}/k')=O(n^{3/4} \log{n})$ that hits all the sets $N_{k'}(v)$ for $v \in G'$. Note that all the nodes already computed the sets $N_{k'}(v)$. The complexity is $O((\log{\log{n}})^3)$ rounds.
Each node updates all the nodes if it is in $A'$ which takes one round. To implement Line \ref{ualg_mssp'}, we use a sparse variant of our MSSP algorithm, as follows. First, we compute a $(\beta,\epsilon)$-hopset $H$ of size $O(n^{3/2}\log{n})$ for $G'$ using Theorem \ref{hopset_thm}. The complexity is $O(\frac{\log^2{n}}{\epsilon})$ rounds and $\beta = O(\frac{\log{n}}{\epsilon})$. Now, the $\beta$-hop distances in $G' \cup H$ give $(1+\epsilon)$-approximation for the distances in $G'$. Hence, to approximate the distances of all nodes from $A'$, we use our $(S,d,k'')$-algorithm in the graph $G' \cup H$, with $S=A',k''=|A'|=O(n^{3/4}\log{n}),d=\beta$. In $G'$ the maximum degree is $k=\sqrt{n}$, and $H$ is of size $O(n^{3/2}\log{n})$, hence the graph $G' \cup H$ is of size $O(n^{3/2}\log{n})$. Thus, by Theorem \ref{thrm:source-detection}, the complexity is $O \Big( \Big( \frac{(n^{3/2}\log{n})^{1/3} \cdot (n^{3/4}\log{n})^{2/3}}{n}+1 \Big) \frac{\log{n}}{\epsilon} \Big) = O(\frac{\log^2{n}}{\epsilon}).$

Implementing Line \ref{ualg_p'} takes one round since each node $v$ knows $N_{k'}(v)$ and $A'$. To implement Line \ref{ualg_dist_p'} it is enough that each node $v$ sends to each node $u$ the values $\{\delta(v,p'(v)),\delta(v,p'(u))\}.$ To conclude the description of the algorithm, we explain how we implement Line \ref{ualg_3paths}. For this we multiply the following 3 matrices. The matrix $M_1(u,w) = \delta(u,w)$ if $w \in N_{k'}(u)$ and $M_1(u,w) = \infty$ otherwise. The matrix $M_3 = M_1^{T}$. I.e., $M_3(w,u)=M_1(u,w)$, and the matrix $M_2(u,v)=\delta(u,v)$ if $\{u,v\} \in E'$ and $M_2(u,v) = \infty$ otherwise. Note that $M_1,M_3$ have $O(n^{5/4})$ entries which are not $\infty$, since the sizes of the sets $N_{k'}(v)$ are $n^{1/4}$. Also, from the definition of $G'$, the matrix $M_2$ has $O(n^{3/2})$ entries which are not $\infty$. We multiply $M_1 \cdot M_2 \cdot M_3$ in the min-plus semiring. For this, we first multiply $M_1 \cdot M_2$ which takes $O(1)$ rounds based on their sparsity, by Theorem \ref{theorem:mm}. Then, we multiply the resulting matrix and $M_3$. Since the maximum degree in $M_1$ is $k' = n^{1/4}$ and the maximum degree in $M_2$ is $k = \sqrt{n}$, we get that in $M_1 \cdot M_2$ the maximum degree is $n^{3/4}$. Hence, the matrix $M_1 \cdot M_2$ has at most $n^{7/4}$ non-$\infty$ elements, and $M_3$ has at most $n^{5/4}$ non-$\infty$ elements. Hence multiplying $(M_1 \cdot M_2) \cdot M_3$ takes $O(1)$ rounds By Theorem \ref{theorem:mm} (choosing $\densh{ST}=n$). Since the multiplication is in the min-plus semiring and by the definition of the matrices, it holds that after the multiplication any pair of nodes $u,v$ computed the value $$\min \{\delta(u,u')+\delta(u',v')+\delta(v',v)| {u' \in N_{k'}(u),v' \in N_{k'}(v), \{u',v'\} \in E'}\}.$$ This completes the proof.
\end{proof}

\subsubsection{Correctness}

\begin{lemma} \label{APSP_approx}
By the end of the algorithm, for all $u,v \in V$, it holds that $\delta(u,v) \leq (2+\epsilon)d_G(u,v).$
\end{lemma}

\begin{proof}
If $d_G(u,v)=1$, then after Line \ref{ualg_init}, $\delta(u,v)=1$ and we are done. Hence, we assume that $d_G(u,v) \geq 2.$
Let $P$ be a shortest path between $u$ and $v$ in $G$. We first consider the case that $P$ has a high-degree node $u'$ of degree at least $k=\widetilde{O}(\sqrt{n})$. Then, $u'$ has a neighbour $w \in A$. For now, we use the notation $d(u,v)$ for $d_G(u,v)$. From the triangle inequality, $d(u,w) \leq d(u,u')+d(u',w)=d(u,u')+1$ and $d(v,w) \leq d(v,u')+1.$ Also, since $u' \in P$, it holds that $d(u,v)=d(u,u')+d(u',v).$ Hence, $d(u,w)+d(w,v) \leq d(u,u')+1+d(u',v)+1 \leq d(u,v)+2 \leq 2d(u,v).$ In Line \ref{ualg_mssp}, we compute $(1+\epsilon)$-approximations for the values $d(u,w),d(w,v)$, hence after Line \ref{ualg_A}, we have $\delta(u,v) \leq (2+2\epsilon)d(u,v).$

We next consider the case that $P$ has only low-degree nodes, hence it is contained in the graph $G'$. From now on we work with the graph $G'$, and use $d(u,v)$ for $d_{G'}(u,v)$. Since $P$ is contained in $G'$, it holds that $d_{G'}(u,v)=d_{G}(u,v)$. The proof is divided to 3 cases, as in the proof of Lemma \ref{wAPSP_approx}.

\textbf{Case 1:} there is a node $w \in P \cap (N_{k'}(u) \cap N_{k'}(v))$. Then, $d(u,v) = d(u,w)+d(w,v).$ Also, $u$ and $v$ computed the values $d(u,w),d(w,v)$ in Line \ref{ualg_k'nearest}, and $\delta(u,v)=d(u,w)+d(w,v)=d_{G'}(u,v)=d_G(u,v)$ after Line \ref{ualg_w}.

\textbf{Case 2:} there is a node $w \in P \setminus (N_{k'}(u) \cup N_{k'}(v))$.
Since $w \in P$, $p'(u) \in N_{k'}(u)$ and $p'(v) \in N_{k'}(v)$, we have $d(u,p'(u)) + d(v,p'(v)) \leq d(u,w) + d(v,w) \leq d(u,v)$.
Exactly as in the proof of Case 2 in Lemma \ref{wAPSP_approx}, we get that $\min \{ d(u,p'(u))+d(p'(u),v), d(v,p'(v))+d(p'(v),u) \}\leq 2d(u,v).$
Since our algorithm computes $d(u,p'(u)),d(v,p'(v))$ exactly in Line \ref{ualg_k'nearest} and approximates $d(u,p'(v)),d(v,p'(u))$ within $(1+\epsilon)$ factor in Line \ref{ualg_mssp'}, after Line \ref{ualg_dist_p'} we have $\delta(u,v) \leq (2+2\epsilon)d(u,v).$

\textbf{Case 3:} $P$ is composed of a path between $u$ to a node $u' \in N_{k'}(u)$, an edge $\{u',v'\} \in E'$ where $v' \in N_{k'}(v)$ and a path between $v'$ and $v$. In this case, $\delta(u,u') = d(u,u')$ and $\delta(v',v)= d(v,v')$ after Line \ref{ualg_k'nearest}, and $\delta(u',v')=d(u,v)$ after Line \ref{ualg_init}. Hence, in Line \ref{ualg_3paths} we set $\delta(u,v) = \delta(u,u')+\delta(u',v')+\delta(v',v)=d(u,v).$

To conclude, in all cases we have $\delta(u,v) \leq (2+2\epsilon)d_G(u,v)$ by the end of the algorithm. Choosing $\epsilon' = 2\epsilon$ completes the proof.
\end{proof}

\subsubsection{Conclusion}

From Lemmas \ref{APSP_time} and \ref{APSP_approx}, we get the following.

\begin{theorem}
There is a deterministic $(2+\epsilon)$-approximation algorithm for \emph{unweighted} APSP in the \clique model that takes $O(\frac{\log^2{n}}{\epsilon})$ rounds.
\end{theorem}

\section{Other applications}

\subsection{Exact single-source shortest paths} \label{section:SSSP}

We show that our distance tools allow to get a faster algorithm for exact SSSP in undirected weighted graphs, improving the previous $\widetilde{O}(n^{1/3})$-round algorithm \cite{censor2015algebraic}.
The algorithm is very simple, and is based on ideas used in \cite{nanongkai2014distributed, elkin2017distributed, shi1999time}. We start by computing the distances to the nodes in $N_k(v)$ for all $v \in V$ using our $k$-nearest algorithm. Then, we add to the graph all the edges $\{\{u,v\}| u \in V, v \in N_k(u) \}$ with the weights $d(u,v)$ computed by the algorithm, which results in a new graph $G'$. $G'$ is called the $k$-shortcut graph or $k$-shortcut hopset in \cite{nanongkai2014distributed, elkin2017distributed}, and it is proved in \cite{nanongkai2014distributed} that the shortest path diameter of $G'$ is $O(n/k)$. I.e., for any pair of nodes $u,v$ there is a shortest path between $u$ and $v$ of hop-distance at most $O(n/k)$ in $G'$.

\begin{lemma} (Theorem 3.10 in \cite{nanongkai2014distributed})
The shortest path diameter of $G'$ is smaller than $4n/k$.
\end{lemma}

Now, to compute SSSP from a certain node we just run the classic Bellman-Ford algorithm \cite{bellman1958routing,ford1956network} in the graph $G'$ which takes $O(n/k)$ rounds in a graph with shortest path diameter $O(n/k)$, even in the more restricted CONGEST model (see e.g., \cite{nanongkai2014distributed}).
To optimize the complexity, we choose $k=n^{5/6}.$
Now, finding the $k$-nearest takes $\widetilde{O}(k/n^{2/3})=\widetilde{O}(n^{1/6})$ rounds by Theorem~\ref{thrm:knearest}, and the Bellman-Ford exploration takes $O(n^{1/6})$ rounds, proving the following.

\begin{theorem}
There is a deterministic algorithm for exact SSSP in undirected weighted graphs that takes $\widetilde{O}(n^{1/6})$ rounds in the \clique model.
\end{theorem}

\subsection{Diameter} \label{section:diameter}

We show that using our distance tools we can implement efficiently the algorithm of Roditty and Vassilevska Williams for approximating the diameter \cite{roditty2013fast}, which is an extension of the algorithm of Aingworth et al. \cite{aingworth1999fast}.

The algorithm for approximating the diameter is based on computing BFS trees from $\widetilde{O}(\sqrt{n})$ nodes. A difference in our case is that we compute $(1+\epsilon)$-approximate distances from the same nodes, and not the exact distances, which affects slightly the approximation obtained. Also, our algorithm works only for \emph{undirected} graphs. Another difference from \cite{roditty2013fast} is that we compute distances to the nodes in $N_k(v)$. This is useful in our case, since we can compute these distances \emph{exactly} which is useful for the analysis, it also allows us to give a deterministic algorithm. The algorithm works as follows.

\begin{oframed}
\begin{enumerate}
\item Each node $v$ computes the distances to $N_k(v)$ for $k = \widetilde{O}(\sqrt{n})$.\label{diam_nearest}
\item We compute a hitting set $S$ of size $\widetilde{O}(\sqrt{n})$ that hits all the sets $\{N_k(v)\}_{v \in V}$.\label{diam_hit}
\item We compute $(1+\epsilon)$-approximate distances from $S$ to all the nodes.\label{diam_S}
\item Let $p(v) \in S \cap N_k(v)$ be a closest node to $v$ from $S$. $v$ knows the distance to $p(v)$ and sends to all the nodes the value $d(v,p(v))$.\label{diam_p}
\item Let $w \in V$ be a node such that $d(w,p(w)) \geq d(v,p(v))$ for all $v \in V$, we compute $(1+\epsilon)$-approximate distances from the set $N_k(w)$ (including $w$) to all the nodes.\label{diam_Nw}
\item The estimate for the diameter is the maximum distance found in Line \ref{diam_S} or \ref{diam_Nw}.\label{diam_end}
\end{enumerate}
\end{oframed}

\begin{claim}
The algorithm takes $O(\frac{\log^2{n}}{\epsilon})$ rounds.
\end{claim}

\begin{proof}
Let $k=\sqrt{n}$. Computing the distances to the $k$-nearest takes $O(\log^2{n})$ rounds by Theorem~\ref{thrm:knearest}. We compute a hitting set $A$ of size $O(\sqrt{n} \log{n})$ that hits the sets $N_k(v)$ using Lemma \ref{det_hit}, and then each node tells all nodes if it is in $S$ which takes 1 round. Line \ref{diam_S} is implemented using our MSSP algorithm that takes $O(\frac{\log^2{n}}{\epsilon})$ rounds since $A$ is of size $O(\sqrt{n} \log{n})$. The distance $d(v,p(v))$ is already known to $v$ since $p(v) \in N_k(v),$ and sending the values $d(v,p(v))$ takes one round. After Line \ref{diam_p}, all the nodes can deduce the node $w$ (breaking ties by ids), since $w$ knows the set $N_k(w)$, it can update all the nodes in $N_k(w)$ that they are in this set, and in one additional round all the nodes learn $N_k(w).$ Then, implementing Line \ref{diam_Nw} is done using our MSSP algorithm which takes $O(\frac{\log^2{n}}{\epsilon})$ rounds since $k=\sqrt{n}$. Computing the maximum estimate takes one additional round. To conclude, the complexity is $O(\frac{\log^2{n}}{\epsilon})$ rounds.
\end{proof}

We next prove that the approximation returned by the algorithm is nearly $3/2$, following the proof in \cite{roditty2013fast}. Since we compute $(1+\epsilon)$-approximate distances, we may also get an estimate that is greater than $D$ by a $(1+\epsilon)$ factor. If we want the estimate to be always at most $D$ we can divide by $1+\epsilon$.

\begin{claim}
Let $G=(V,E)$ be an undirected graph with diameter $D=3h+z$, where $h \geq 0$ and $z \in \{0,1,2\}$, and let $D'$ be the estimate returned by the algorithm. Then, $2h+z \leq D' \leq (1+\epsilon)D$ if $z \in \{0,1\}$, and $2h+1 \leq D' \leq (1+\epsilon)D$ if $z=2$.
\end{claim}

\begin{proof}
The analysis is based on the analysis in \cite{roditty2013fast} (see Lemma 4). The difference in our algorithm compared to \cite{roditty2013fast} is that in some cases we approximate distances instead of computing exact ones. However, all the distances $d(v,p(v))$ are computed exactly, which shows that the definition of $w$ in our algorithm is exactly as in \cite{roditty2013fast}. Let $a,b$ be such that $d(a,b)=D$, where $D$ is the diameter of the graph. The analysis in \cite{roditty2013fast} is divided to cases. In the first case, $d(w,p(w)) \leq h$, and they show that in this case there is a node at distance at least $2h+z$ from $p(a)$. Our algorithm computes $(1+\epsilon)$-approximate distances to all the nodes $p(v)$, hence at least one of our estimates is at least $2h+z$ in this case.

A second case is that there is a node of distance at least $2h+z$ from $w$. Since we compute approximate distances from $w$, one of our estimates is at least $2h+z$ in this case. The last case to consider is that all the nodes are at distance smaller than $2h+z$ from $w$, and $d(w,p(w)) > h$. In this case, they show that there exists a node $w' \in N_k(w)$ such that $d(a,w') \geq 2h+1$. Since we compute $(1+\epsilon)$-approximate distances from all the nodes in $N_k(w)$, our estimate is at least $2h+1.$
In addition, since all the distances computed are $(1+\epsilon)$-approximate distances, our estimate is at most $(1+\epsilon)D$. This completes the proof.
\end{proof}

\textbf{Remark:} as stated in \cite{roditty2013fast}, the same analysis works also for weighted graphs with non-negative weights with a loss of an additive $W$ term where $W$ is the maximum edge-weight. I.e., the estimate obtained satisfies $\lfloor 2D/3-W \rfloor \leq D' \leq (1+\epsilon)D.$ The implementation is exactly the same in our case.


\section{Discussion}
\label{section:discussion}
In this paper, we provide new tools for distance computations in the \clique model based on sparse MM. We demonstrate the power of these tools by providing efficient algorithms for many central problems, such as APSP, SSSP, and diameter approximation. We believe that these techniques may be useful for many additional tasks in distance computation in the \clique, and possibly also in additional related models. Many intriguing questions are still open.

First, while we provide constant approximations for many tasks in polylogarithmic time, the complexity of our exact SSSP algorithm is still polynomial. It would be interesting to study whether it is possible to obtain a polylogarithmic algorithm for the exact case. Surprisingly, even for the simple task of computing a BFS tree there is currently no sub-polynomial algorithm in the \clique.

Second, while our MM algorithms work also for the directed case, our hopset construction is only for the undirected case and hence our results are for undirected graphs. Also, as we explain in Section \ref{sec:cont}, obtaining \emph{any} approximation for directed APSP in sub-polynomial time would give a sub-polynomial algorithm for matrix multiplication. It would be interesting to study whether it is possible to provide sub-polynomial algorithms for directed SSSP.

Our unweighted APSP algorithm provides a $(2+\epsilon)$-approximation, whereas in the weighted case there is currently an additional additive term, and a natural question is whether we can get a $(2+\epsilon)$-approximation for weighted APSP in polylogarithmic time as well. Another interesting goal is to obtain \emph{additive} approximations in sub-polynomial time.

\paragraph{Acknowledgements:} We thank Mohsen Ghaffari, Michael Elkin and Merav Parter for fruitful discussions. This project has received funding from the European Union's Horizon 2020 Research And Innovation Program under grant agreement no.\ 755839.

\bibliography{References}


\end{document}